\documentclass[a4paper,USenglish,cleveref,thm-restate]{lipics-template/lipics-v2021}

\usepackage[algo2e, linesnumbered,ruled]{algorithm2e}

\usepackage{hyperref}

\usepackage{booktabs}
\usepackage[table]{xcolor}
\usepackage{siunitx}
\usepackage{tikz}
\usetikzlibrary{calc, decorations.pathreplacing}

\usepackage{caption}
\usepackage{subcaption}

\usepackage{amsmath}
\usepackage{amsfonts}
\usepackage{amssymb}

\usepackage[capitalise]{cleveref}

\usepackage{xspace}

\graphicspath{{./figures/}}

\hideLIPIcs

\title{Uniform Generation of Temporal Graphs with Given Degrees}

\author{Daniel Allendorf}{Goethe University Frankfurt, Germany}{daniel@ae.cs.uni-frankfurt.de}{https://orcid.org/0000-0002-0549-7576}{}

\authorrunning{D.~Allendorf}

\Copyright{Daniel Allendorf}



\def\switch#1{\mathsf{#1}}

\begin{document}

\maketitle

\begin{abstract}
	Uniform sampling from the set $\mathcal{G}(\mathbf{d})$ of graphs with a given degree-sequence $\mathbf{d} = (d_1, \dots, d_n) \in \mathbb N^n$ is a classical problem in the study of random graphs.
	We consider an analogue for temporal graphs in which the edges are labeled with integer timestamps.
	The input to this generation problem is a tuple $\mathbf{D} = (\mathbf{d}, T) \in \mathbb N^n \times \mathbb N_{>0}$ and the task is to output a uniform random sample from the set $\mathcal{G}(\mathbf{D})$ of temporal graphs with degree-sequence $\mathbf{d}$ and timestamps in the interval $[1, T]$.
	By allowing repeated edges with distinct timestamps, $\mathcal{G}(\mathbf{D})$ can be non-empty even if $\mathcal{G}(\mathbf{d})$ is, and as a consequence, existing algorithms are difficult to apply.
	
	We describe an algorithm for this generation problem which runs in expected linear time $O(M)$ if $\Delta^{2+\epsilon} = O(M)$ for some constant $\epsilon > 0$ and $T - \Delta = \Omega(T)$ where $M = \sum_i d_i$ and $\Delta = \max_i d_i$.
	Our algorithm applies the switching method of McKay and Wormald \cite{DBLP:journals/jal/McKayW90} to temporal graphs: we first generate a random temporal \emph{multigraph} and then remove self-loops and duplicated edges with switching operations which rewire the edges in a degree-preserving manner.
\end{abstract}

%\clearpage

%\setcounter{page}{1}

\section{Introduction}
\label{sec:introduction}
A common problem in network science is the sampling of a graph matching a given degree-sequence.
Formally, given a sequence of integers $\mathbf{d} = (d_1, \dots, d_n)$, we say that a graph $G = (V, E)$ with nodes $V = \{v_1, \dots, v_n\}$ matches $\mathbf{d}$ if the number of incident edges at node $v_i$ equals $d_i$ for each $1 \leq i \leq n$.
We then define $\mathcal{G}(\mathbf{d})$ as the set of all simple graphs (e.g. without loops or multi-edges) matching $\mathbf{d}$ and ask for a uniform random sample $G \in \mathcal{G}(\mathbf{d})$.
Such a sample is useful as it allows us to construct null models for testing the influence of the degrees on other graph properties of interest~\cite{gotelli1996null,DBLP:journals/im/BlitzsteinD11}.
In addition, this sampling problem is tightly related to the task of estimating $|\mathcal{G}(\mathbf{d})|$~\cite{DBLP:journals/iandc/SinclairJ89,DBLP:journals/jal/McKayW90}.

\emph{Temporal graphs} are capable of modeling not only the topology but also the time structure of networks (see \cite{holme2012temporal} or \cite{doi:10.1080/17445760.2012.668546} for an overview).
Possibly the most common type of temporal graph augments each edge of a classical graph with an integer timestamp.
Here, we work by the following definition.
\begin{definition}[Temporal Graph]
	\label{def:temporal-graph}
	A temporal (multi-)graph $G = (V, E)$ consists of a set of nodes $V = \{v_1, \dots, v_n\}$ and a (multi-)set of edges $E = \{e_1, \dots, e_m\}$ where each edge is a tuple $(\{u, v\}, t) \in \{ \{u, v\} : u, v \in V \} \times \mathbb N_{>0}$.
\end{definition}
In terms of semantics, the presence of an edge $(\{u, v\}, t)$ indicates that the nodes $u$ and $v$ are connected at time $t$.
For the purpose of modeling networks, it additionally makes sense to restrict ourselves to \emph{simple} temporal graphs which exclude certain types of edges.
To this end, we call a temporal graph $G = (V, E)$ \emph{simple} if the edge set $E$ contains no loops and no edges between the same nodes with the same timestamp, i.e. iff $u \neq v$ for all $(\{u, v\}, t) \in E$ and $E$ is a set.

Given its prominence in classical graph theory, it is reasonable to assume that parameterizing temporal graphs by the degrees can yield a similarly useful model.
For instance, the distribution of active times of nodes and edges in real-world temporal graphs has been observed to follow a power-law \cite{barabasi2005origin,HARDER2006329,holme2012temporal}, which can be reproduced by sampling a uniform random temporal graph with power-law degrees.
Still, to the best of our knowledge, the problem of generating such graphs has not been considered so far.
In this paper, we study algorithmic techniques of sampling simple temporal graphs uniformly at random, and in particular, focus on the task of providing an exact uniform sample.

\subsection{Related Work}

Conditions for the realizability of a sequence of integers $\mathbf{d} = (d_1, \dots, d_n)$ as a graph were given by Havel~\cite{havel1955remark} and Hakimi~\cite{hakimi1962realizability}.
While the provided proofs can be used to construct a realization $G \in \mathcal{G}(\mathbf{d})$, the resulting graph is deterministic.

A simple way to obtain a uniform random sample $G \in \mathcal{G}(\mathbf{d})$ is to use the \emph{configuration model} of Bender and Canfield~\cite{DBLP:journals/jct/BenderC78} or Bollob{\'a}s~\cite{DBLP:journals/ejc/Bollobas80} to sample random \emph{multigraphs} with sequence $\mathbf{d}$ until a simple graph is found.
Unfortunately, this simple rejection scheme is not efficient as its run time is exponential in the largest degree $\Delta = \max_i d_i$~\cite{DBLP:journals/jal/Wormald84}.

Efficient algorithms have been obtained by using the \emph{switching} method of McKay and Wormald~\cite{DBLP:journals/jal/McKayW90} which arose from the use of switchings for enumeration in McKay~\cite{mckay1985asymptotics}.
The approach is to again start from a random multigraph but instead of rejecting non-simple graphs outright, loops and multi-edges are removed with switching operations which rewire the edges while preserving the degrees of the nodes.
In addition to an algorithm with expected runtime $O(m)$ for generating graphs with $m$ edges and bounded degrees $\Delta^4 = O(m)$~\cite{DBLP:journals/jal/McKayW90, DBLP:journals/rsa/ArmanGW21}, efficient algorithms which use this method have been given for generating $d$-regular graphs in expected time $O(n d + d^4)$ if $d = o(\sqrt{n})$~\cite{DBLP:journals/siamcomp/GaoW17, DBLP:journals/rsa/ArmanGW21}, and graphs with power-law degrees in expected time $O(n)$ if the exponent satisfies $\gamma > (21 + \sqrt{61}) / 10$~\cite{DBLP:conf/soda/GaoW18, DBLP:journals/rsa/ArmanGW21}.

Finally, there exist efficient solutions to various relaxations of the problem.
For instance, we may allow the graph to match the sequence only in expectation~\cite{chung2002connected, DBLP:conf/waw/MillerH11}, or use a Markov chain to approximate the uniform distribution~\cite{DBLP:journals/tcs/JerrumS90, rao1996, DBLP:journals/corr/Zhao13b, DBLP:journals/ejc/ErdosGMMSS22}.
See also~\cite{greenhillsurvey} for a survey of the most relevant techniques and results.

\subsection{Our Contribution}

We give results on sampling temporal graphs with given degrees and lifetime.
Formally, given a tuple $\mathbf{D} = (\mathbf{d}, T)$, we say that a temporal multigraph $G$ with nodes $V = \{v_1, \dots, v_n\}$ matches $\mathbf{D}$ if the sum of the numbers of incident edges at node $v_i$ over all $T$ timestamps equals $d_i$ for each $1 \leq i \leq n$.
If at least one simple graph matches $\mathbf{D}$, we call $\mathbf{D}$ \emph{realizable}.
The temporal graph generation problem now asks to output a sample $G \in \mathcal{G}(\mathbf{D})$ uniformly at random from the set $\mathcal{G}(\mathbf{D})$ of matching simple temporal graphs.
Note that by allowing repeated connections, a given tuple $\mathbf{D} = (\mathbf{d}, T)$ can be realizable as a simple temporal graph even if $\mathbf{d}$ is not realizable as a classical simple graph.
More severely, consider the sequence  $\mathbf{d}_t, 1 \leq t \leq T$ of degree-sequences induced by the individual timestamps, then there exist sequences which satisfy $\sum_t \mathbf{d}_t = \mathbf{d}$ but are not realizable as a sequence of simple graphs due to loops or multi-edges which cannot be rewired with switchings which preserve $\mathbf{d}_t, 1 \leq t \leq T$.
In other words, even distributing the degrees among the timestamps is not trivial, and existing algorithms are difficult to apply to sampling temporal graphs.
Instead, switchings are required which operate on the temporal graph as a whole and re-assign timestamps where necessary, and this is the key feature of the algorithm which we describe here.
This algorithm, called \textsc{T-Gen}, generates simple temporal graphs with bounded degrees.
Our main result is as follows (see \autoref{subsec:thm2-proof} for proof details).
\begin{theorem}
	\label{thm:t-gen}
	Given a realizable tuple $\mathbf{D} = (\mathbf{d}, T)$ which satisfies $\Delta^{2+\epsilon} = O(M)$ for a constant $\epsilon > 0$ and $T - \Delta = \Omega(T)$, \textsc{T-Gen} outputs a uniform random sample $G \in \mathcal{G}(\mathbf{D})$ in expected time $O(M)$.
\end{theorem}
As customary, we assume an underlying sequence of tuples $(\mathbf{D}_M)_{M \in \mathbb N}$ and give the asymptotic runtime as $M \to \infty$.
In particular, the conditions on the input tuple can be understood as being imposed on functions $\Delta(M)$, $T(M)$.

\subsection{Overview}
\label{subsec:techniques}

The general idea of \textsc{T-Gen} is to apply the switching method of~\cite{DBLP:journals/jal/McKayW90} to temporal graphs.
To this end, we define a \emph{temporal configuration model} (see \autoref{sec:temporal-cm}) which samples a random temporal multigraph with the property that the probability of a given graph only depends on the contained loops and temporal multi-edges, i.e. multiple edges between the same nodes and with the same timestamp.
In particular, a simple temporal graph output by this model has the uniform distribution.
Usually, the obtained graph is not simple, but if the input tuple $\mathbf{D}$ satisfies the conditions imposed in \autoref{thm:t-gen}, then the number of non-simple edges is sufficiently small to allow for efficient removal via switchings.
An interesting property of the random model in this context is that the expected number of loops is not affected by the desired lifetime $T$, whereas the expected number of temporal double-edges scales with $O(1 / T)$ (see \autoref{lem:tcm-temporal-loops-multi-edges}).
This allows for a looser bound on the degrees in terms of the number of edges (e.g. $\Delta^{2+\epsilon} = O(M)$ as opposed to $\Delta^{4} = O(m)$ for classical graphs with general bounded degrees).
As discussed, a challenging aspect of generating simple temporal graphs is that the degree-sequences $\mathbf{d}_t, 1 \leq t \leq T$ implied by a random temporal multigraph may not be realizable as a sequence of classical simple graphs.
This necessitates switchings which rewire edges across different time slices of the graph and re-assign timestamps (see \autoref{def:tl-switching} and \autoref{fig:switching-TL} for an example).
As a consequence, the number of timestamps we can assign to an edge without creating a temporal multi-edge affects the distribution of the graphs, and to preserve uniformity, it becomes necessary to account for the available timestamps.
To discuss this matter, we briefly describe the technique used in~\cite{DBLP:journals/jal/McKayW90} to correct the distribution.

Generally speaking, when analyzing a switching operation $\theta$ we fix subsets $\mathcal{S}, \mathcal{S'} \subseteq \mathcal{M}(\mathbf{d})$ of the set $\mathcal{M}(\mathbf{d})$ of multigraphs matching the sequence $\mathbf{d}$.
Considering the edges rewired by $\theta$ then associates each graph $G \in \mathcal{S}$ with a subset $\mathcal{F}(G) \subseteq \mathcal{S'}$ of graphs in $\mathcal{S'}$ which can be produced by performing a type $\theta$ switching on $G$, and each graph $G' \in \mathcal{S'}$ with a subset $\mathcal{B}(G') \subseteq \mathcal{S}$ of graphs on which we can perform a type $\theta$ switching which produces $G'$.
Given this setup, the goal is to start from a uniform random graph $G \in \mathcal{S}$ and perform a type $\theta$ switching to obtain a uniform random graph $G' \in \mathcal{S'}$.
To this end, let $f(G) = |\mathcal{F}(G)|$, $b(G') = |\mathcal{B}(G')|$, and assume that $\mathcal{F}(G)\neq \emptyset$ for every $G \in \mathcal{S}$ and $\mathcal{B}(G') \neq \emptyset$ for every $G' \in \mathcal{S'}$.
Then, if we start from a graph $G$ uniformly distributed in $\mathcal{S}$, and perform a uniform random type $\theta$ switching on $G$, the probability of producing a given graph $G' \in \mathcal{S'}$ is
\begin{equation}
	\nonumber
	\hfill \sum_{G \in \mathcal{B}(G')} \frac{1}{|\mathcal{S}| f(G)} \hfill
\end{equation}
which depends on $G'$ if $\mathcal{F}(G)$ and $\mathcal{B}(G')$ vary over different choices of $G$ and $G'$.
To correct this, \emph{rejection} steps can be used which restart the algorithm with a certain probability.
Before performing the switching, we \emph{f-reject} (forward reject) with probability $1 - f(G) / \overline{f}(\mathcal{S})$ where $\overline{f}(\mathcal{S})$ is an upper bound on $f(G)$ over all graphs $G \in \mathcal{S}$, and after performing the switching, we \emph{b-reject} (backward reject) with probability $1 - \underline{b}(\mathcal{S'}) / b(G')$ where $\underline{b}(\mathcal{S'})$ is a lower bound on $b(G')$ over all graphs $G' \in \mathcal{S'}$.
The probability of producing $G'$ is now
\begin{equation}
	\nonumber
	\hfill \sum_{G \in \mathcal{B}(G')} \frac{1}{|\mathcal{S}| f(G)} \frac{f(G)}{\overline{f}(\mathcal{S})} \frac{\underline{b}(\mathcal{S'})}{b(G')} = \frac{b(G')}{|\mathcal{S}| \overline{f}(\mathcal{S})} \frac{\underline{b}(\mathcal{S'})}{b(G')} = \frac{\underline{b}(\mathcal{S'})}{|\mathcal{S}| \overline{f}(\mathcal{S})} \hfill
\end{equation}
which only depends on $\mathcal{S}$ and $\mathcal{S'}$, implying that $G'$ has the uniform distribution if $G$ does.

Still, correcting the distribution in this way is efficient only if the typical values of $f(G)$ and $b(G')$ do not deviate too much from $\overline{f}(\mathcal{S})$ and $\underline{b}(\mathcal{S'})$.
To avoid a high probability of restarting in cases where this does not hold, Gao and Wormald~\cite{DBLP:journals/siamcomp/GaoW17} first used the idea of using additional switchings which partially equalize the probabilities by mapping graphs of higher probability to graphs of lower probability. 
This is implemented via a Markov chain which either chooses a main switching to remove a non-simple edge or an additional switching to equalize the probabilities.

\textsc{T-Gen} similarly uses additional switchings but without the use of a Markov chain.
Instead, we always perform the main kind of switching first and then an additional switching which targets a specified set of edges involved in the main kind of switching performed.
The issue we address in this way is that the typical number of available timestamps for an edge is $\Omega(T)$, whereas the corresponding lower bound must be set to $T - (\Delta - 1)$ due to the possibility of a graph in which the edge has multiplicity $\Delta - 1$.
This would imply a high probability of restarting unless we place a harsher restriction on $\Delta$.
Fortunately, the conditions imposed in \autoref{thm:t-gen} suffice to ensure that the highest multiplicity of any edge in the initial graph is bound by a constant $\eta = O(1)$ with high probability (see \autoref{lem:tcm-kappa-lambda-eta}, \autoref{sec:temporal-cm}), and we can also show that any edges created by the algorithm itself only increase the multiplicity up to a constant $\mu = O(1)$ (\autoref{lem:edge-bound-probability}, \autoref{sec:t-gen}).
Now, after performing a main kind of switching, we partition the subset $\mathcal{S'}$ which contains the obtained graph into the subsets $\mathcal{S'}_{\mathbf{m}<\mu}$ and $\mathcal{S'} \setminus \mathcal{S'}_{\mathbf{m} < \mu}$ where $\mathcal{S'}_{\mathbf{m}<\mu}$ contains all graphs in which all specified edges have multiplicity less than $\mu$.
We then equalize the probabilities of producing the graphs in $\mathcal{S'}_{\mathbf{m}<\mu}$ via switchings which involve the specified edges with the standard rejection step (which is efficient by $\mu = O(1)$), and reset the probability of graphs in $\mathcal{S'} \setminus \mathcal{S'}_{\mathbf{m} < \mu}$ to zero by rejecting these graphs.
To equalize the probabilities between the graphs in $\mathcal{S'}_{\mathbf{m}<\mu}$ and $\mathcal{S'} \setminus \mathcal{S'}_{\mathbf{m} < \mu}$, we define auxiliary switching operations which map the graphs in $\mathcal{S'}_{\mathbf{m}<\mu}$ to graphs in $\mathcal{S'} \setminus \mathcal{S'}_{\mathbf{m} < \mu}$ and an identity switching which maps any graph in $\mathcal{S'}_{\mathbf{m}<\mu}$ to itself, and specify a probability distribution over these two kinds of switchings which ensures that all graphs in $\mathcal{S'}$ are produced with the same probability via switchings which involve the specified edges.
It then only remains to equalize the probabilities over different choices of the specified edges via standard rejection steps which finally results in the obtained graph being uniformly distributed in $\mathcal{S'}$.

\section{Temporal Configuration Model}
\label{sec:temporal-cm}
The temporal configuration model samples a random temporal multigraph matching a given tuple $\textbf{D} = (\mathbf{d}, T) \in \mathbb N^n \times \mathbb N_{>0}$ (provided that $M = \sum_i d_i$ is even).
It can be implemented as follows.
First, for each node index $i \in \{1, \dots, n\}$, put $d_i$ marbles labeled $i$ into an urn.
Then, starting from the empty graph $G = (V, \emptyset)$ on the node set $V = \{v_1, \dots, v_n\}$, add edges by iteratively performing the following steps until the urn is empty:
\begin{enumerate}
	\item Draw two marbles from the urn uniformly at random (without replacement), and let $i, j$ denote the labels of those marbles.
	\item Draw a timestamp $t$ uniformly at random from the set of timestamps $[1, T]$.
	\item Add the temporal edge $(\{v_i, v_j\}, t)$ to the graph $G$.
\end{enumerate}

%\noindent
In the following, we analyze the output distribution of this random model.
To this end, we first give some definitions to characterize the edges in a temporal multigraph.
Given two nodes $v_i, v_j \in V$ and a timestamp $t \in [1, T]$, we define $w_{i,j,t}$ as the number of edges between $v_i$ and $v_j$ with timestamp $t$ in the graph, and call $w_{i,j,t}$ the \emph{temporal multiplicity} of the edge $(\{v_i, v_j\}, t)$.
Then, if $w_{i,j,t} \geq 2$, we say that the edge is contained in a \emph{temporal multi-edge}, and in the special cases $w_{i,j,t} = 2$ and $w_{i,j,t} = 3$, refer to the multi-edge as a \emph{double-edge} and \emph{triple-edge}, respectively.
In addition, we define $m_{i,j} = \sum_t w_{i,j,t}$ as the total number of edges between $v_i$ and $v_j$ over all timestamps and call $m_{i,j}$ the \emph{multiplicity} of $\{v_i, v_j\}$.
Finally, we call an edge $(\{v_i\}, t)$ which connects a node $v_i$ to itself a \emph{loop} at $v_i$, and in the cases where $w_{i,t} = 1$ and $w_{i,t} = 2$, refer to the edge as a \emph{temporal single-loop} and \emph{temporal double-loop}, respectively.

Now, for a given temporal multigraph $G$, let $\mathbf{W}(G)$ denote the $n \times n \times T$ tensor such that the entries $\mathbf{W}_{i,j,t}(G)$ where $i \neq j$ contain $w_{i,j,t}$ if $w_{i,j,t} \geq 2$ and $0$ if otherwise, and the entries $\mathbf{W}_{i,i,t}(G)$ contain $w_{i,t}$.
In other words, $\mathbf{W}(G)$ specifies the temporal multiplicities of all temporal multi-edges and loops in $G$.
In addition, let $\mathcal{M}(\mathbf{D})$ denote the set of temporal multigraphs matching a given tuple $\mathbf{D}$, and for a given tensor $\mathbf{W}$, let $\mathcal{S}(\mathbf{W})$ denote the subset of temporal multigraphs $G \in \mathcal{M}(\mathbf{D})$ such that $\mathbf{W}(G) = \mathbf{W}$.
Then, the following holds.
\begin{theorem}
	\label{thm:tcm-output-distribution}
	Let $G$ be a temporal multigraph output by the temporal configuration model on an input tuple $\mathbf{D}$.
	Then, $G$ is uniformly distributed in the set $\mathcal{S}(\mathbf{W}(G)) \subseteq \mathcal{M}(\mathbf{D})$.
\end{theorem}

\begin{proof}
	For a given tuple $\textbf{D} = (\mathbf{d}, T)$, let $P_i = \{i_1, \dots, i_{d_i}\}$ where $1 \leq i \leq n$, $P = \bigcup_{1\leq i \leq n} P_i$, and define a \emph{temporal configuration} of $\textbf{D}$ as a partition of the set $P$ into $M / 2$ subsets of size two in which each subset is assigned an integer $t \in [1, T]$.
	Then, the temporal configuration model samples a temporal configuration uniformly at random and outputs the corresponding temporal multigraph (by identifying the nodes corresponding to the labels and replacing the subsets together with the assigned integers by temporal edges).
	Thus, the probability of a given graph $G$ is proportional to the number of temporal configurations corresponding to $G$, which equals the number of ways to label the edges in $G$ with the labels in $P$ which give a distinct temporal configuration.
	Denote this number by $C_P(G)$, and observe that if $H_1 = (V, E_1)$, $H_2 = (V, E_2)$ are the subgraphs of the non-simple edges and simple edges in $G = (V, E)$, respectively, then as $E_1 \cap E_2 = \emptyset$ and $E_1 \cup E_2 = E$, we have $C_P(G) = C_{P_1}(H_1) C_{P_2}(H_2) \prod_i \binom{d_i}{k_i}$ where $\mathbf{k}$ denotes the degree-sequence of $H_2$ (arbitrarily), $P_1 = \bigcup_i \{i_1, \dots, i_{d_i - k_i}\}$ and $P_2 = \bigcup_i \{i_1, \dots, i_{k_i}\}$.
	In addition, we have $C_{P_2}(H_2) = \prod_i k_i !$ as all edges incident at the nodes in the simple temporal graph $H_2$ are distinct, which implies that all possible labelings result in a distinct temporal configuration. 
	Finally, observe that $\mathbf{W}(G)$ determines both $H_1$ and $\mathbf{k}$, and thus $C_P(G)$ only depends on $\mathbf{W}(G)$.
\end{proof}

\noindent
Note that the set of \emph{simple} temporal graphs matching $\mathbf{D}$ corresponds to the special case $\mathcal{G}(\mathbf{D}) = \mathcal{S}(\mathbf{0}^{n \times n \times T}) \subseteq \mathcal{M}(\mathbf{D})$.
Thus, \autoref{thm:tcm-output-distribution} implies that a simple temporal graph output by the random model is uniformly distributed in the set $\mathcal{G}(\mathbf{D})$.
In general, the probability of obtaining a simple temporal graph is small (see \autoref{lem:tcm-temporal-loops-multi-edges} below).
Still, there are conditions under which the numbers and multiplicities of non-simple edges are manageable.
We state these conditions in terms of the following properties of a degree-sequence $\mathbf{d}$:
\begin{equation}
	\nonumber
	\hfill \Delta = \max_{1 \leq i \leq n} d_i, \hfill M = \sum_{1 \leq i \leq n} d_i, \hfill M_2 = \sum_{\substack{1 \leq i \leq n \\ d_i > 0}} d_i (d_i - 1). \hfill
\end{equation}

\begin{lemma}
	\label{lem:tcm-temporal-loops-multi-edges}
	Let $\mathbf{D} = (\mathbf{d}, T)$ be a tuple which satisfies $\Delta^{2} = o(M)$ and $\Delta = O(T)$, and $G$ a graph output by the temporal configuration model when given $\mathbf{D}$ as input.
	Then, the expected number of temporal double-edges in $G$ is $O(M_2^2 / M^2 T)$, the expected number of temporal single-loops is $O(M_2 / M)$, and with high probability there are no temporal double-loops or temporal triple-edges.
\end{lemma}

\begin{proof}
	It is straightforward to check that the probability that $G$ contains $m$ given temporal edges is $O(M^{-m} T^{-m})$.
	Thus, the expected number of temporal double-edges in a graph output by the temporal configuration model is
	\begin{equation}
		\nonumber
		\hfill O\left(\sum_{1 \leq i, j \leq n}  \sum_{1 \leq t \leq T} \frac{4 \binom{d_i}{2} \binom{d_j}{2}}{M^{2} T^{2}}\right) = O\left( \frac{M_2^2}{M^2 T} \right), \hfill
	\end{equation}
	the expected number of temporal single-loops is
	\begin{equation}
		\nonumber
		\hfill O\left(\sum_{1 \leq i \leq n} \sum_{1 \leq t \leq T} \frac{\binom{d_i}{2}}{M T}\right) = O\left( \frac{M_2}{M} \right), \hfill
	\end{equation}
	the expected number of temporal triple-edges is
	\begin{equation}
		\nonumber
		\hfill O\left(\sum_{1 \leq i, j \leq n} \sum_{1 \leq t \leq T} \frac{6 \binom{d_i}{3} \binom{d_j}{3}}{M^3 T^3}\right) = O\left( \frac{\Delta^2 M_2^2}{M^3 T^2} \right) = O\left(\frac{\Delta^2}{M}\right) = o(1), \hfill
	\end{equation}
	and the expected number of temporal double-loops is
	\begin{equation}
		\nonumber
		\hfill O\left(\sum_{1 \leq i \leq n} \sum_{1 \leq t \leq T} \frac{3 \binom{d_i}{4}}{M^2 T^2}\right) = O\left( \frac{\Delta^2 M_2}{M^2 T} \right) = O\left(\frac{\Delta^2}{M}\right) = o(1). \hfill \qedhere \qed
	\end{equation}
\end{proof}

\begin{lemma}
	\label{lem:tcm-kappa-lambda-eta}
	Let $\mathbf{D} = (\mathbf{d}, T)$ be a tuple which satisfies $\Delta^{2+\epsilon} = O(M)$ for a constant $\epsilon > 0$ and $\Delta = O(T)$, and $G$ a graph output by the temporal configuration model when given $\mathbf{D}$ as input.
	Then, with high probability, the number of incident temporal double-edges at any node in $G$ is at most $\kappa = \lfloor 1 + 1 / \epsilon \rfloor$, the number of incident temporal single-loops at any node is at most $\lambda = \lfloor 1 + 1 / \epsilon \rfloor$, and the highest multiplicity of any edge is at most $\eta = \lfloor 2 + 2 / \epsilon \rfloor$.
\end{lemma}

\begin{proof}
	Define $\delta = \frac{\epsilon}{2+\epsilon}$ and note that $\Delta^{2+\epsilon} = O(M) \implies \Delta = O(M^{\delta / \epsilon})$, and if $\epsilon > 0$, then $\Delta^2 / M = O(M^{-\delta}) = o(1)$.
	Let $K_m$ denote the number of nodes incident with $m$ temporal double-edges in a graph output by the temporal configuration model.
	Then
	\begin{align}
		\nonumber
		\quad \quad \quad \quad \mathbb{E}[K_m] &= O\left( \sum_{1 \leq i \leq n} \sum_{1 \leq j_1, \dots, j_m \leq n} \sum_{1 \leq t_1, \dots, t_m \leq T} \frac{\binom{2m}{2, \dots, 2} \binom{d_i}{2m} \prod_{k=1}^m \binom{d_{j_k}}{2}}{M^{2m} T^{2m}} \right)
		\\
		\nonumber
		&= O\left(\frac{M_{2m} M_2^{m}}{M^{2m} T^m} \right) = O\left(\frac{\Delta^{2m - 1}}{M^{m-1}} \right).
	\end{align}
	where $M_k = \sum_{1 \leq i \leq n} \prod_{1 \leq j \leq k} (d_i - j + 1)$ and the last equality follows by $M_{k+1} < \Delta M_k$.
	Thus, if $\Delta^{2+\epsilon} = O(M)$ for some constant $\epsilon > 0$, the probability of at least one node incident with more than $\kappa = \lfloor 1 + 1 / \epsilon \rfloor$ temporal double-edges is at most
	\begin{align}
		\nonumber
		\quad \quad \sum_{m= \kappa+1 }^{\lfloor \Delta / 2 \rfloor} \mathbb{E}[K_m] &= O\left( \sum_{m= \kappa+1 }^{\lfloor \Delta / 2 \rfloor} \frac{\Delta^{2m - 1}}{M^{m-1}} \right) 
		\\
		\nonumber
		&= O\left( \sum_{m= \kappa+1 }^{\lfloor \Delta / 2 \rfloor} M^{- \delta (m - 1 - 1 / \epsilon)} \right) 
		\\
		\nonumber
		&= O\left( \sum_{m= \kappa+2 }^{\lfloor \Delta / 2 \rfloor} M^{- \delta (m - 1 - 1 / \epsilon)} \right)  + O\left(M^{- \delta (\underbrace{\kappa + 1}_{=\lfloor 1 + 1 / \epsilon \rfloor + 1 > 1 + 1 / \epsilon} - 1 - 1 / \epsilon)}\right)
		\\
		\nonumber
		&= O\left( \sum_{i=1}^{\lfloor \Delta / 2 \rfloor -  (\kappa + 1)} M^{- \delta i} \right) + o(1)
		\\
		\nonumber
		&= O\left( M^{-\delta} \right) + o(1)
		\\
		\nonumber
		&= o(1).
	\end{align}
	Similarly, let $L_m$ denote the number of nodes incident with $m$ temporal single-loops, then
	\begin{equation}
		\nonumber
		\hfill \mathbb{E}[L_m] = O\left( \sum_{1 \leq i \leq n} \sum_{1 \leq t_1, \dots, t_m \leq T} \frac{\binom{2m}{2, \dots, 2} \binom{d_i}{2m}}{M^{m} T^{m}} \right) = O\left(\frac{M_{2m}}{M^{m}} \right) = O\left(\frac{\Delta^{2m-1}}{M^{m-1}} \right) \hfill
	\end{equation}
	which if $\Delta^{2+\epsilon} = O(M)$ implies that the probability of at least one node incident with more than $\lambda = \lfloor 1 + 1 / \epsilon \rfloor$ temporal single-loops is $o(1)$ by the same argument as above.
	Finally, let $H_m$ denote the number of ordinary multi-edges of multiplicity $m$, then
	\begin{equation}
		\nonumber
		\hfill \mathbb{E}[H_m] = O\left( \sum_{1 \leq i, j \leq n} \frac{m! \binom{d_i}{m} \binom{d_j}{m}}{M^m} \right) = O\left(\frac{M_m^2}{M^{m}} \right) = O\left(\frac{\Delta^{2 m - 2}}{M^{m-2}} \right) \hfill
	\end{equation}
	which if $\Delta^{2+\epsilon} = O(M)$ implies that the probability of at least one ordinary multi-edge of multiplicity higher than $\eta = \lfloor 2 + 2 / \epsilon \rfloor$ is $o(1)$ by the same argument as above.
\end{proof}

\section{Algorithm T-Gen}
\label{sec:t-gen}
\textsc{T-Gen} takes a realizable tuple $\mathbf{D}$ as input and outputs a uniform random sample $G \in \mathcal{G}(\mathbf{D})$ from the set of matching simple temporal graphs. 
The algorithm starts by sampling a random temporal multigraph $G \in \mathcal{M}(\mathbf{D})$ via the temporal configuration model (see \autoref{sec:temporal-cm}).
It then checks if $G$ satisfies initial conditions on the numbers and multiplicities of non-simple edges (see \autoref{subsec:initial-rejection}).
In particular, the initial graph $G$ is not allowed to contain temporal triple-edges, or temporal double-loops (or any higher multiplicities).
If $G$ satisfies these conditions, then the algorithm proceeds to removing all temporal single-loops and temporal double-edges during two stages.
Stage 1 (\autoref{subsec:stage1}) removes all temporal single-loops in the graph.
For this purpose three kinds of switching operations are used.
The main kind of switching removes a temporal single-loop at a specified node and with a specified timestamp.
After performing this kind of switching we always perform one of two auxiliary switchings.
The purpose of these switchings is to equalize the probabilities between graphs which contain ordinary multi-edges of high multiplicity and graphs which do not.
Stage 2 (\autoref{subsec:stage2}) removes all temporal  double-edges, i.e. double-edges which share the same timestamp.
Doing this efficiently requires five kinds of switchings, two of which remove a temporal double-edge between two specified nodes and with a specified timestamp, and three of which are auxiliary switchings.
Once all non-simple edges have been removed, the resulting simple temporal graph is output.

\subsection{Initial Conditions}
\label{subsec:initial-rejection}

The initial conditions for the random multigraph $G$ are as follows.
Define
\begin{equation}
	\nonumber
	\hfill B_L = \frac{M_2}{M}, \quad \quad \quad B_D = \frac{M_2^2}{M^2 T}, \hfill
\end{equation}
let $L = \sum_{i, t} \mathbf{W}_{i,i,t}(G)$ and $D = \sum_{i \neq j, t} \mathbf{W}_{i,j,t}(G)$ denote the sums of the multiplicities of loops and temporal multi-edges of $G$, respectively, and choose three constants
\begin{equation}
	\nonumber
	\hfill \lambda \geq 1 + 1 / \epsilon, \quad \quad \quad \quad \kappa \geq 1 + 1 / \epsilon, \quad \quad \quad \quad \mu \geq 3 + 2 / \epsilon \hfill
\end{equation}
where $\epsilon > 0$ is a constant such that $\Delta^{2+\epsilon} = O(M)$ (or set $\lambda = \kappa = \mu = \Delta$ if no such constant exists).
Then, $G$ satisfies the initial conditions if $L \leq B_L$, $D / 2 \leq B_D$, there are no loops of temporal multiplicity $w \geq 2$ or temporal multi-edges of temporal multiplicity $w \geq 3$, and no node is incident with more than $\lambda$ temporal single-loops or $\kappa$ temporal double-edges.

Observe that the lower bounds on the constants $\lambda$ and $\kappa$ are implied by \autoref{lem:tcm-kappa-lambda-eta}.
The lower bound on $\mu$ is implied by the following result (proof in \autoref{apx-subsec:run-time-proofs}).
\begin{lemma}
	\label{lem:edge-bound-probability}
	If the input tuple $\mathbf{D} = (\mathbf{d}, T)$ satisfies $\Delta^{2+\epsilon} = O(M)$ for a constant $\epsilon > 0$ and $T - \Delta = \Omega(T)$, then with high probability none of the graphs visited during a given run of \textsc{T-Gen} contain an edge of multiplicity higher than $\mu = \lfloor 3 + 2 / \epsilon \rfloor$.
\end{lemma}
As a simple temporal graph is allowed to contain ordinary multi-edges, the constant $\mu$ cannot be enforced by rejecting violating initial graphs.
Instead, we equalize the probabilities between graphs which contain such edges and graphs which do not via the auxiliary switchings mentioned above.
We describe this approach in detail in \autoref{subsec:stage1} and \autoref{subsec:stage2}.

Finally, there are two special cases.
First, if the initial multigraph $G$ is simple, then this graph can be output without checking the preconditions or going through any of the stages.
Second, if the input tuple does not satisfy 
\begin{equation}
	\nonumber
	\hfill M > 16 \Delta^2 + 4 \Delta + 2 B_L + 4 B_D, \quad \quad \quad \quad T > \Delta - 1 \hfill
\end{equation}
then \textsc{T-Gen} restarts until a simple graph is found and output.
Note that these requirements are satisfied if $\Delta^{2+\epsilon} = O(M)$ and $T - \Delta = \Omega(T)$, however, we make them explicit here to ensure correctness on any input tuple.

In all other cases, \textsc{T-Gen} restarts if the graph $G$ does not satisfy the initial conditions.
Otherwise, the algorithm enters Stage 1 to remove the temporal single-loops in $G$.

\subsection{Stage 1: Removal of Temporal Single-Loops}
\label{subsec:stage1}

\begin{figure}[t]
	\centering
	\resizebox{0.6\textwidth}{!}{
		\def\node#1{}
\def\thickness{thick}
\def\saturation{15}
\begin{tikzpicture}[
	point/.style={
		draw,
		circle,
		inner sep=0,
		fill,
		minimum width=0.35em,
		minimum height=0.35em
	},
	vertex/.style={
		draw,
		circle,
		inner sep=0.25em,
		\thickness,
		fill=white
	},
	edge/.style={
		draw, 
		black, 
		solid,
		\thickness
	}]
	
	\filldraw[fill=red!\saturation!white, draw=black!0!black, \thickness, rounded corners] (-5.25em,5.5em) rectangle (-0.75em,1em);
	\node (t1) at (-4.75em, 5em) {\footnotesize{$t_1$}};
	\node[vertex] (1v1) at (-3em, 2.5em) {$v_1$};
	\draw[edge] (-3.75em, 3em) .. controls (-5em, 5.5em) and (-1.0em, 5.5em) .. (-2.25em, 3em);
	
	\filldraw[fill=green!\saturation!white, draw=black!0!black, \thickness, rounded corners] (-7.5em,0em) rectangle (-3.5em,-6em);
	\node (t2) at (-7em, -0.5em) {\footnotesize{$t_2$}};
	\node[vertex] (1v2) at (-5.5em, -1.5em) {$v_2$};
	\node[vertex] (1v4) at (-5.5em, -4.5em) {$v_4$};
	\path[edge] (1v2) -- (1v4);
	
	\filldraw[fill=green!\saturation!white, draw=black!0!black, \thickness, rounded corners] (-2.5em,0em) rectangle (1.5em,-6em);
	\node (t3) at (-2em, -0.5em) {\footnotesize{$t_3$}};
	\node[vertex] (1v3) at (-0.5em, -1.5em) {$v_3$};
	\node[vertex] (1v5) at (-0.5em, -4.5em) {$v_5$};
	\path[edge] (1v3) -- (1v5);
	
	\path[->, thick, draw] (3.5em,  -0.25em) -- (6.5em,  -0.25em) {};
	
	\filldraw[fill=green!\saturation!white, draw=black!0!black, \thickness, rounded corners] (8.5em,5.5em) rectangle (12.5em,-2em);
	\node (t4) at (9em, 5em) {\footnotesize{$t_4$}};
	\node[vertex] (2v1a) at (10.5em, 4em) {$v_1$};
	\node[vertex] (2v2) at (10.5em, 0em) {$v_2$};
	\path[edge] (2v1a) -- (2v2);
	
	\filldraw[fill=green!\saturation!white, draw=black!0!black, \thickness, rounded corners] (13.5em,5.5em) rectangle (17.5em,-2em);
	\node (t5) at (14em, 5em) {\footnotesize{$t_5$}};
	\node[vertex] (2v1b) at (15.5em, 4em) {$v_1$};
	\node[vertex] (2v3) at (15.5em, 0em) {$v_3$};
	\path[edge] (2v1b) -- (2v3);
	
	\filldraw[fill=green!\saturation!white, draw=black!0!black, \thickness, rounded corners] (8.5em,-3em) rectangle (17.5em,-6em);
	\node (t6) at (9em, -3.5em) {\footnotesize{$t_6$}};
	\node[vertex] (2v4) at (10.5em, -4.5em) {$v_4$};
	\node[vertex] (2v5) at (15.5em, -4.5em) {$v_5$};
	\path[edge] (2v4) -- (2v5);
\end{tikzpicture}
	}
	\caption{The $\switch{TL}$ switching removes a temporal single-loop with timestamp $t_1$ at a node $v_1$. A shaded region labeled with a timestamp contains all edges with this timestamp between the nodes. Red and green shades indicate non-simple and simple edges, respectively.}
	\label{fig:switching-TL}
\end{figure}
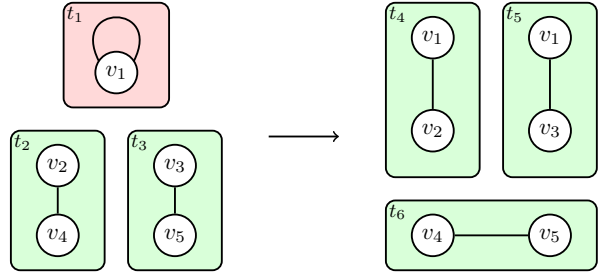

Stage 1 removes the temporal single-loops in the graph.
Doing this efficiently requires multiple kinds of switching operations.
In total, we use three kinds of switchings which we denote as $\switch{TL}$, $\switch{A}_{m,n}$ and $\switch{I}$.
We formally define each kind of switching below.
The switching of the main kind written as $\switch{TL}$ removes a temporal single-loop at a specified node and with a specified timestamp.
After performing this kind of switching, an $\switch{A}_{m,n}$ auxiliary switching is performed with a certain probability.
This switching adds up to two edges with multiplicities $\max \{m, n\} \geq \mu$ to the graph to equalize the probability of producing graphs with or without these kinds of edges.
In addition, we define the identity switching $\switch{I}$ which maps each graph to itself.
This is done to formalize the event in which no auxiliary switching is performed.
Formal definitions of the $\switch{TL}$ and $\switch{A}_{m,n}$ switchings are as follows.

\begin{definition}[$\mathsf{TL}$ switching at $v_1, t_1$]
	\label{def:tl-switching}
	For a graph $G$ such that $(\{v_1\}, t_1)$ is a temporal single-loop, let $(\{v_2, v_4\}, t_2)$, $(\{v_3, v_5\}, t_3)$ be edges and $t_4, t_5, t_6 \in [1, T]$ timestamps such that
	\begin{itemize}
		\item none of the edges $(\{v_2, v_4\}, t_2)$, $(\{v_3, v_5\}, t_3)$ is a loop or in a temporal double-edge,
		\item the nodes $v_2, v_3, v_4, v_5$ are distinct from $v_1$, and $v_4$ is distinct from $v_5$, and
		\item none of the edges $(\{v_1, v_2\}, t_4)$, $(\{v_1, v_3\}, t_5)$, $(\{v_4, v_5\}, t_6)$ exist.
	\end{itemize}
	Then, a $\mathsf{TL}$ switching replaces the edges $(\{v_1\}, t_1)$, $(\{v_2, v_4\}, t_2)$, $(\{v_3, v_5\}, t_3)$ with $(\{v_1, v_2\}, t_4)$, $(\{v_1, v_3\}, t_5)$, $(\{v_4, v_5\}, t_6)$ (see \autoref{fig:switching-TL}).
\end{definition}
We stress that the integer subscripts of nodes are used in place of generic indices to reduce visual clutter and simplify the descriptions.
Naturally, these labels may still refer to any node in the graph.

\begin{definition}[$\mathsf{A}_{m,n}$ switching at $v_1, v_2, v_3, v_4, v_5$]
	\label{def:amn-switching}
	For a graph $G$ such that $\{v_2, v_4\}$ and  $\{v_3, v_5\}$ are non-edges, let $(\{v_2, v_{2i + 4} \}, t_i)$, $(\{v_4, v_{2i + 5} \}, t_{m + i})$, $1 \leq i \leq m$ be incident edges at $v_2$, $v_4$, $(\{v_3, v_{2m + 2i + 4} \}, t_{2m + i})$, $(\{v_5, v_{2m + 2i + 5} \}, t_{2m + n + i})$, $1 \leq i \leq n$ incident edges at $v_3$, $v_5$, and $t_{2m+2n+1}, \dots, t_{4m+4n} \in [1, T]$ timestamps such that
	\begin{itemize}
		\item none of the edges $(\{v_2, v_{2i + 4} \}, t_i)$, $(\{v_4, v_{2i + 5} \}, t_{m + i})$, $1 \leq i \leq m$, $(\{v_3, v_{2m + 2i + 4} \}, t_{2m + i})$ and none of $(\{v_3, v_{2m + 2i + 5} \}, t_{2m + n + i})$, $1 \leq i \leq n$ is a loop or in a temporal double-edge,
		\item the nodes $v_1, \dots, v_{2m + 2n + 5}$ are all distinct, and
		\item none of the edges $(\{v_2, v_4\}, t_{2m + 2n + i})$, $(\{v_{2i + 4}, v_{2i + 5}\}, t_{3m + 2n + i})$, $1 \leq i \leq m$ and none of $(\{v_3, v_5\}, t_{4m + 2n + i})$, $(\{v_{2m + 2i + 4}, v_{2m + 2i + 5}\}, t_{4m + 3n + i})$, $1 \leq i \leq n$ exist.
	\end{itemize}
	Then, an $\mathsf{A}_{m,n}$ switching replaces the edges $(\{v_2, v_{2i + 4} \}, t_i)$, $(\{v_4, v_{2i + 5} \}, t_{m + i})$, $1 \leq i \leq m$, $(\{v_3, v_{2m + 2i + 4} \}, t_{2m + i})$, $(\{v_3, v_{2m + 2i + 5} \}, t_{2m + n + i})$, $1 \leq i \leq n$ with $(\{v_2, v_4\}, t_{2m + 2n + i})$, $(\{v_{2i + 4}, v_{2i + 5}\}, t_{3m + 2n + i})$, $1 \leq i \leq m$, $(\{v_3, v_5\}, t_{4m + 2n + i})$, $(\{v_{2m + 2i + 4}, v_{2m + 2i + 5}\}, t_{4m + 3n + i})$, $1 \leq i \leq n$.
\end{definition}
In other words, the $\mathsf{TL}$ switching chooses two edges and then rewires the specified loop and the two edges such that exactly the specified loop is removed and no other non-simple edges are created or removed.
Likewise, the $\switch{A}_{m,n}$ switching chooses $m$ incident edges at two nodes $v_2, v_4$ each and $n$ incident edges at two nodes $v_3, v_5$ each and then rewires the edges such that exactly $m$ simple edges between the nodes $v_2, v_4$ and exactly $n$ simple edges between the nodes $v_3, v_5$ are created and no non-simple edges are created or removed.

After each $\switch{TL}$ switching, we perform an $\switch{A}_{m,n}$ auxiliary switching or the identity switching.
To decide which switching to perform we define a probability distribution over the different types of switchings which ensures uniformity.
In total, the set of $\switch{A}_{m,n}$ auxiliary switchings is
\begin{equation}
	\nonumber
	\hfill \Theta_\switch{A} = \bigcup_{\substack{0 \leq m,n < \Delta \\ \mu \leq \max\{m, n\} }} \{\switch{A}_{m,n}\}. \hfill.
\end{equation}
The switching to be performed is then sampled from the distribution $(\Theta_\switch{A} \cup \{ \switch{I} \}, P_\switch{A})$ where
\begin{equation}
	\nonumber
	\hfill p_\switch{A}(\switch{A}_{m,n}) = p_\switch{A}(\switch{I}) \frac{\overline{f}_{\switch{A}_{m,n}}(\mathbf{W'})}{\underline{b}_{\switch{A}_{m,n}}(\mathbf{W'})}, \quad \quad \quad p_\switch{A}(\switch{I}) = 1 - \sum_{\theta \in \Theta_\switch{A}} p_\switch{A}(\theta) \hfill
\end{equation}
for quantities $\overline{f}_{\switch{A}_{m,n}}(\mathbf{W'})$ and $\underline{b}_{\switch{A}_{m,n}}(\mathbf{W'})$ given further below.

On a high level, Stage 1 runs in a loop until a rejection occurs or all temporal single-loops have been removed from $G$.
As our switchings remove a temporal single-loop at a specified node and with a specified timestamp, the temporal single-loops can be removed in an arbitrary order.
To this end, let $\pi$ denote a permutation of the entries in $\mathbf{W}(G)$ such that $\mathbf{W}_{i,i,t} = 1$.
Then, Stage 1 iterates through the temporal single-loops in the order given by $\pi$ and performs the following steps for each temporal single-loop.

\begin{enumerate}
	\item Let $G$ denote the current graph, $\mathbf{W} = \mathbf{W}(G)$ and $(\{v_1\}, t_1)$ the loop.
	\item Pick a uniform random $\switch{TL}$ switching $S$ which removes $(\{v_1\}, t_1)$ from $G$.
	\item Restart (\textbf{f-reject}) with probability $1 - \frac{f_\switch{TL}(G)}{\overline{f}_\switch{TL}(\mathbf{W})}$.
	\item Rewire the edges according to $S$, let $G'$ denote the resulting graph and $\mathbf{W'} = \mathbf{W}(G')$.
	\item Let $(\{v_2, v_4\}, t_2)$, $(\{v_3, v_5\}, t_3)$ denote the edges removed by $S$. 
	\item Restart if $m_{2,4} \geq \mu$ or $m_{3,5} \geq \mu$.
	\item Restart (\textbf{b-reject}) with probability $1 -  \frac{\underline{b}_\switch{TL}(\mathbf{W'}; 2)}{b_\switch{TL}(G', v_1 v_2 v_3 v_4 v_5; 2)}$.
	\item Choose a switching type $\theta \sim (\Theta_\switch{A} \cup \{ \switch{I} \}, P_\switch{A})$.
	\item If $\theta = \switch{A}_{m,n}$ for some $\switch{A}_{m,n} \in \Theta_\switch{A}$:
	\begin{enumerate}
		\item Restart if $m_{2,4} \geq 1$ or $m_{3,5} \geq 1$.
		\item Pick a uniform random $\switch{A}_{m,n}$ switching $S'$ which adds an edge with node set $\{v_2, v_4\}$ and multiplicity $m$ and an edge with node set $\{v_3, v_5\}$ and multiplicity $n$ to $G'$.
		\item Restart (\textbf{f-reject}) with probability $1 - \frac{f_{\switch{A}_{m,n}}(G')}{\overline{f}_{\switch{A}_{m,n}}(\mathbf{W'})}$.
		\item Rewire the edges according to $S'$ and let $G''$ denote the resulting graph.
		\item Restart (\textbf{b-reject}) with probability $1 - \frac{\underline{b}_{\switch{A}_{m,n}}(\mathbf{W'})}{b_{\switch{A}_{m,n}}(G'', v_1 v_2 v_3 v_4 v_5)}$.
		\item Set $G' \gets G''$.
	\end{enumerate}
	\item Restart (\textbf{b-reject}) with probability $1 -  \frac{\underline{b}_\switch{TL}(\mathbf{W'}; 0)\underline{b}_\switch{TL}(\mathbf{W'}; 1)}{b_\switch{TL}(G', v_1; 0)b_\switch{TL}(G', v_1 v_2 v_3; 1)}$.
	\item Set $G \gets G'$.
\end{enumerate}

\noindent
To fully specify Stage 1, it remains to define the quantities used for the f- and b-rejection steps.
For the f-rejection in step $3$, define $f_{\switch{TL}}(G)$ as the number of $\switch{TL}$ switchings which can be performed on the graph $G$.
The corresponding upper bound is
\begin{equation}
	\nonumber
	\hfill \overline{f}_{\switch{TL}}(\mathbf{W}) = M^2 T^3. \hfill
\end{equation}

\noindent
For the b-rejections in steps $7$ and $10$, define $b_\switch{TL}(G', v_1 v_2 v_3 v_4 v_5; 2)$ as the number of timestamps $t_2, t_3 \in [1, T]$ such that the edges $(\{v_2, v_4\}, t_2)$, $(\{v_3, v_5\}, t_3)$ do not exist in $G'$, $b_\switch{TL}(G', v_1 v_2 v_3; 1)$ as the number of simple temporal edges $(\{v_4, v_5\}, t_6)$ such that $v_4$, $v_5$ are distinct from $v_1$, $v_2$, $v_3$, and $b_\switch{TL}(G', v_1; 0)$ as the number of distinct simple temporal edges $(\{v_1, v_2\}, t_4)$, $(\{v_1, v_3\}, t_5)$ incident at $v_1$.
The lower bounds on these quantities are
\begin{equation}
	\nonumber
	\hfill \underline{b}_\switch{TL}(\mathbf{W'}; 2) = (T - (\mu - 1))^2, \quad \underline{b}_\switch{TL}(\mathbf{W'}; 1) = M - 2 B_L - 4 B_D - 4 \Delta, \quad \underline{b}_\switch{TL}(\mathbf{W'}; 0) = k_1 (k_1 - 1) \hfill
\end{equation}
where $k_i = d_i - \sum_{t} (2 \mathbf{W'}_{i,i,t} + \sum_{1 \leq j \leq n : j \neq i} \mathbf{W'}_{i,j,t})$ denotes the number of simple edges at $v_i$.

For the f-rejection in step $9c$, define $f_{\switch{A}_{m,n}}(G')$ as the number of $\switch{A}_{m,n}$ switchings which can be performed on the graph $G'$.
The upper bound is
\begin{equation}
	\nonumber
	\hfill \overline{f}_{\switch{A}_{m,n}}(\mathbf{W'}) = \Delta^{2 (m+n)} T^{2 (m+n)}. \hfill
\end{equation}

\noindent
For the b-rejection in step $9e$, define $b_{\switch{A}_{m,n}}(G'', v_1 v_2 v_3 v_4 v_5)$ as the number of $\switch{A}_{m,n}$ switchings which can produce the graph $G''$.
The corresponding lower bound is
\begin{equation}
	\nonumber
	\hfill \underline{b}_{\switch{A}_{m,n}}(\mathbf{W'}) = (M - 2 B_L - 4 B_D - 4 (m+n+3) \Delta)^{m+n} (T - (\Delta - 1))^{2 (m+n)}. \hfill
\end{equation}

\noindent
The correctness of Stage 1 is implied by the following result, the proof of which can be found in \autoref{apx-subsec:uniformity-proofs}.

\begin{lemma}
	\label{lem:st1-uniformity}
	The graph $G'$ at the end of an iteration of Stage 1 is uniformly distributed in $\mathcal{S}(\mathbf{W'})$ given that the graph $G$ at the start of the iteration is uniformly distributed in $\mathcal{S}(\mathbf{W})$.
\end{lemma}

\noindent
Showing that Stage 1 is efficient requires showing that both the probability of restarting and the run time of all iterations is small.
To this end, we show the following in \autoref{apx-subsec:run-time-proofs}. 

\begin{lemma}
	\label{lem:st1-rejection}
	The probability of not restarting in Stage 1 is $\exp(- O(\Delta^2 / M) - O(\Delta / T))$.
\end{lemma}

\begin{lemma}
	\label{lem:st1-runtime}
	The expected run time of Stage 1 is $O(\Delta^2)$.
\end{lemma}

\noindent
Stage 1 ends if all temporal single-loops have been removed.
The algorithm then moves on to Stage 2 to remove the remaining temporal double-edges.

\subsection{Stage 2: Removal of Temporal Double-Edges}
\label{subsec:stage2}

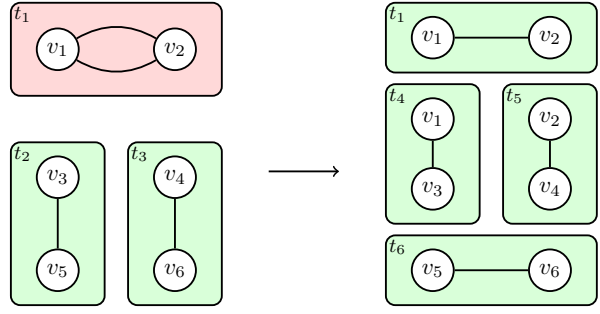
\begin{figure}[t]
	\centering
	\resizebox{0.6\textwidth}{!}{
		\def\node#1{}
\def\thickness{thick}
\def\saturation{15}
\begin{tikzpicture}[
	point/.style={
		draw,
		circle,
		inner sep=0,
		fill,
		minimum width=0.35em,
		minimum height=0.35em
	},
	vertex/.style={
		draw,
		circle,
		inner sep=0.25em,
		\thickness,
		fill=white
	},
	edge/.style={
		draw, 
		black, 
		solid,
		\thickness
	}]
	
	\filldraw[fill=red!\saturation!white, draw=black!0!black, \thickness, rounded corners] (-7.5em, 7em) rectangle (1.5em, 3em);
	\node (t1) at (-7em, 6.5em) {\footnotesize{$t_1$}};
	\node[vertex] (1v1) at (-5.5em, 5em) {$v_1$};
	\node[vertex] (1v2) at (-0.5em, 5em) {$v_2$};
	\draw[edge] (1v1) to[bend right] (1v2);
	\draw[edge] (1v1) to[bend left] (1v2);
	
	\filldraw[fill=green!\saturation!white, draw=black!0!black, \thickness, rounded corners] (-7.5em,1em) rectangle (-3.5em,-6em);
	\node (t2) at (-7em, 0.5em) {\footnotesize{$t_2$}};
	\node[vertex] (1v3) at (-5.5em, -0.5em) {$v_3$};
	\node[vertex] (1v5) at (-5.5em, -4.5em) {$v_5$};
	\path[edge] (1v3) -- (1v5);
	
	\filldraw[fill=green!\saturation!white, draw=black!0!black, \thickness, rounded corners] (-2.5em,1em) rectangle (1.5em,-6em);
	\node (t3) at (-2em, 0.5em) {\footnotesize{$t_3$}};
	\node[vertex] (1v4) at (-0.5em, -0.5em) {$v_4$};
	\node[vertex] (1v6) at (-0.5em, -4.5em) {$v_6$};
	\path[edge] (1v4) -- (1v6);
	
	\path[->, thick, draw] (3.5em,  -0.25em) -- (6.5em,  -0.25em) {};
	
	\filldraw[fill=green!\saturation!white, draw=black!0!black, \thickness, rounded corners] (8.5em, 7em) rectangle (17.5em,4em);
	\node (t1) at (9em, 6.5em) {\footnotesize{$t_1$}};
	\node[vertex] (2v1) at (10.5em, 5.5em) {$v_1$};
	\node[vertex] (2v2) at (15.5em, 5.5em) {$v_2$};
	\draw[edge] (2v1) -- (2v2);
	
	\filldraw[fill=green!\saturation!white, draw=black!0!black, \thickness, rounded corners] (8.5em, 3.5em) rectangle (12.5em, -2.5em);
	\node (t4) at (9em, 3em) {\footnotesize{$t_4$}};
	\node[vertex] (2v1b) at (10.5em, 2em) {$v_1$};
	\node[vertex] (2v3) at (10.5em, -1em) {$v_3$};
	\path[edge] (2v1b) -- (2v3);
	
	\filldraw[fill=green!\saturation!white, draw=black!0!black, \thickness, rounded corners] (13.5em, 3.5em) rectangle (17.5em, -2.5em);
	\node (t5) at (14em, 3em) {\footnotesize{$t_5$}};
	\node[vertex] (2v2b) at (15.5em, 2em) {$v_2$};
	\node[vertex] (2v4) at (15.5em, -1em) {$v_4$};
	\path[edge] (2v2b) -- (2v4);
	
	\filldraw[fill=green!\saturation!white, draw=black!0!black, \thickness, rounded corners] (8.5em,-3em) rectangle (17.5em,-6em);
	\node (t6) at (9em, -3.5em) {\footnotesize{$t_6$}};
	\node[vertex] (2v5) at (10.5em, -4.5em) {$v_5$};
	\node[vertex] (2v6) at (15.5em, -4.5em) {$v_6$};
	\path[edge] (2v5) -- (2v6);
\end{tikzpicture}
	}
	\caption{The $\switch{TD}_1$ switching removes a temporal double-edge with timestamp $t_1$ between two nodes $v_1, v_2$ while leaving a single-edge between the nodes. A shaded region labeled with a timestamp contains all edges with this timestamp between the nodes. Red and green shades indicate non-simple and simple edges, respectively.}
	\label{fig:switching-TD}
\end{figure}

Stage 2 uses five kinds of switchings which we denote as $\switch{TD}_1$, $\switch{TD}_0$, $\switch{B}_{m, n}$, $\switch{C}_{m,n,o,p}$ and $\switch{I}$.
The two main switchings $\switch{TD}_1$ and $\switch{TD}_0$ remove a temporal double-edge between two specified nodes with a specified timestamp.
The difference is that the $\switch{TD}_1$ switching only removes one occurrence of the edge while the $\switch{TD}_0$ switching erases both occurrences.
This is done to equalize the probability between graphs in which the removed temporal double-edge is a non-edge, or single edge.
After performing a $\switch{TD}_1$ switching, the $\switch{B}_{m, n}$ auxiliary switching may be performed, and after performing a $\switch{TD}_0$ switching, the $\switch{C}_{m, n,o,p}$ switching may be performed.
These auxiliary switchings add edges with multiplicity $\max \{m, n\} \geq \mu$ or $\max \{m, n,o,p\} \geq \mu$ to the graph to equalize the probabilities between graphs with or without these edges. 
We define the $\mathsf{TD}_1$ switching as follows.

\begin{definition}[$\mathsf{TD}_1$ switching at $(\{v_1, v_2\}, t_1)$]
	\label{def:td1-switching}
	For a graph $G$ such that $(\{v_1, v_2\}, t_1)$ is contained in a temporal double-edge, let $(\{v_3, v_5\}, t_2)$, $(\{v_4, v_6\}, t_3)$ be edges and $t_4, t_5, t_6 \in [1, T]$ timestamps such that
	\begin{itemize}
		\item none of the edges $(\{v_3, v_5\}, t_2)$, $(\{v_4, v_6\}, t_3)$ is contained in a temporal double-edge,
		\item the nodes $v_3, v_4, v_5, v_6$ are distinct from $v_1$ and $v_2$, and $v_5$ is distinct from $v_6$, and
		\item none of the edges $(\{v_1, v_3\}, t_4)$, $(\{v_2, v_4\}, t_5)$, $(\{v_5, v_6\}, t_6)$ exist.
	\end{itemize}
	Then, a $\mathsf{TD}_1$ switching replaces the edges $(\{v_1, v_2\}, t_1)$, $(\{v_3, v_5\}, t_2)$, $(\{v_4, v_6\}, t_3)$ with $(\{v_1, v_3\}, t_4)$, $(\{v_2, v_4\}, t_5)$, $(\{v_5, v_6\}, t_6)$ (see \autoref{fig:switching-TD}).
\end{definition}
The $\mathsf{TD}_0$ switching can be defined analogously by using four edges in place of two to remove both edges contained in the temporal double-edge.
To choose a $\switch{TD}_1$ or $\switch{TD}_0$ switching when removing a temporal double-edge, specify the probability distribution $(\{ \switch{TD}_1, \switch{TD}_0 \}, P)$ where
\begin{equation}
	\nonumber
	\hfill p(\switch{TD}_1) = p(\switch{TD}_0) \frac{p_\switch{C}(\switch{I})}{p_\switch{B}(\switch{I})} \frac{\overline{f}_{\switch{TD}_1}(\mathbf{W})\underline{b}_{\switch{TD}_0}(\mathbf{W'})}{ \overline{f}_{\switch{TD}_0}(\mathbf{W}) \underline{b}_{\switch{TD}_1}(\mathbf{W'})}, \quad \quad \quad p(\switch{TD}_0) = 1 - p(\switch{TD}_1) \hfill
\end{equation}
for quantities $\overline{f}_{\switch{TD}_1}(\mathbf{W})$, $\overline{f}_{\switch{TD}_0}(\mathbf{W})$ and $\underline{b}_{\switch{TD}_1}(\mathbf{W'}) = \underline{b}_{\switch{TD}_1}(\mathbf{W'}; 0) \underline{b}_{\switch{TD}_1}(\mathbf{W'}; 1) \underline{b}_{\switch{TD}_1}(\mathbf{W'}; 2)$, $\underline{b}_{\switch{TD}_0}(\mathbf{W'}) = \underline{b}_{\switch{TD}_0}(\mathbf{W'}; 0) \underline{b}_{\switch{TD}_0}(\mathbf{W'}; 1) \underline{b}_{\switch{TD}_0}(\mathbf{W'}; 2)$ defined below and where $p_\switch{B}(\switch{I})$ and $p_\switch{C}(\switch{I})$ are the probabilities of not performing a $\mathsf{B}_{m,n}$ and $\mathsf{C}_{m,n,o,p}$ auxiliary switching, respectively.

Continuing with the auxiliary switchings, we now define the $\mathsf{B}_{m,n}$ switching.
The $\mathsf{C}_{m,n,o,p}$ switching can be defined analogously by expanding the number of edges created from the two edges removed by the $\switch{TD}_1$ switching to the four edges removed by the $\switch{TD}_0$ switching.
\begin{definition}[$\mathsf{B}_{m,n}$ switching at $v_1, v_2, v_3, v_4, v_5, v_6$]
	For a graph $G$ such that $\{v_3, v_5\}$ and  $\{v_4, v_6\}$ are non-edges, let $(\{v_3, v_{2i + 5} \}, t_i)$, $(\{v_5, v_{2i + 6} \}, t_{m + i})$, $1 \leq i \leq m$ be incident edges at $v_3$, $v_5$, $(\{v_4, v_{2m + 2i + 5} \}, t_{2m + i})$, $(\{v_6, v_{2m + 2i + 6} \}, t_{2m + n + i})$, $1 \leq i \leq n$ incident edges at $v_4$, $v_6$, and $t_{2m+2n+1}, \dots, t_{4m+4n} \in [1, T]$ timestamps such that
	\begin{itemize}
		\item none of the edges $(\{v_3, v_{2i + 5} \}, t_i)$, $(\{v_5, v_{2i + 6} \}, t_{m + i})$, $1 \leq i \leq m$, $(\{v_4, v_{2m + 2i + 5} \}, t_{2m + i})$, $(\{v_6, v_{2m + 2i + 6} \}, t_{2m + n + i})$, $1 \leq i \leq n$ is contained in a temporal double-edge,
		\item the nodes $v_1, \dots, v_{2m + 2n + 6}$ are all distinct, and
		\item none of the edges $(\{v_3, v_5\}, t_{2m + 2n + i})$, $(\{v_{2i + 5}, v_{2i + 6}\}, t_{3m + 2n + i})$, $1 \leq i \leq m$ and none of the edges $(\{v_4, v_6\}, t_{4m + 2n + i})$, $(\{v_{2m + 2i + 5}, v_{2m + 2i + 6}\}, t_{4m + 3n + i})$, $1 \leq i \leq n$ exist.
	\end{itemize}
	Then, a $\mathsf{B}_{m,n}$ switching replaces the edges $(\{v_3, v_{2i + 5} \}, t_i)$, $(\{v_5, v_{2i + 6} \}, t_{m + i})$, $1 \leq i \leq m$, $(\{v_4, v_{2m + 2i + 5} \}, t_{2m + i})$, $(\{v_6, v_{2m + 2i + 6} \}, t_{2m + n + i})$, $1 \leq i \leq n$ with $(\{v_3, v_5\}, t_{2m + 2n + i})$, $(\{v_{2i + 5}, v_{2i + 6}\}, t_{3m + 2n + i})$, $1 \leq i \leq m$, $(\{v_4, v_6\}, t_{4m + 2n + i})$, $(\{v_{2m + 2i + 5}, v_{2m + 2i + 6}\}, t_{4m + 3n + i})$, $1 \leq i \leq n$.
\end{definition}

\noindent
The sets of $\switch{B}_{m,n}$ and $\switch{C}_{m,n,o,p}$ switchings are
\begin{equation}
	\nonumber
	\hfill \Theta_\switch{B} = \bigcup_{\substack{0 \leq m,n < \Delta \\ \mu \leq \max\{m,n\} }} \{\switch{B}_{m,n}\} \quad \quad \quad \quad \quad \Theta_\switch{C} = \bigcup_{\substack{0 \leq m,n,o,p < \Delta \\ \mu \leq \max\{m,n,o,p\} }} \{\switch{C}_{m,n,o,p}\} \hfill
\end{equation}
and the associated type distributions are $(\Theta_\switch{B} \cup \{ \switch{I} \}, P_\switch{B})$ and $(\Theta_\switch{C} \cup \{ \switch{I} \}, P_\switch{C})$ where
\begin{align}
	\nonumber
	\hfill \quad \quad \quad \quad \quad p_\switch{B}(\switch{B}_{m,n}) &= p_\switch{B}(\switch{I}) \frac{\overline{f}_{\switch{B}_{m,n}}(\mathbf{W'})}{\underline{b}_{\switch{B}_{m,n}}(\mathbf{W'})}, \quad \quad \quad &p_\switch{B}(\switch{I}) &= 1 - \sum_{\theta \in \Theta_\switch{B}} p_\switch{B}(\theta),& \hfill
	\\
	\nonumber
	\hfill p_\switch{C}(\switch{C}_{m,n,o,p}) &= p_\switch{C}(\switch{I}) \frac{\overline{f}_{\switch{C}_{m,n,o,p}}(\mathbf{W'})}{\underline{b}_{\switch{C}_{m,n,o,p}}(\mathbf{W'})}, \quad \quad \quad &p_\switch{C}(\switch{I}) &= 1 - \sum_{\theta \in \Theta_\switch{C}} p_\switch{C}(\theta).& \hfill
\end{align}

\noindent
The main loop of Stage 2 is as follows.
Let $\pi$ denote a permutation of the entries in $\mathbf{W}(G)$ such that $\mathbf{W}_{i,j,t} = 2$ and $i \neq j$.
Then, Stage 2 iterates through the temporal double-edges in the order given by $\pi$ and performs the following steps.

\begin{enumerate}
	\item Let $G$ denote the current graph, $\mathbf{W} = \mathbf{W}(G)$ and $(\{v_1, v_2\}, t_1)$ the temporal double-edge.
	\item Choose a switching type $\theta \sim (\{ \switch{TD}_0, \switch{TD}_1 \}, P)$.
	\item Pick a uniform random $\theta$ switching $S$ which removes $(\{v_1, v_2\}, t_1)$ from $G$.
	\item Restart (\textbf{f-reject}) with probability $1 - \frac{f_{\theta}(G)}{\overline{f}_{\theta}(\mathbf{W})}$.
	\item Rewire the edges according to $S$, let $G'$ denote the resulting graph and $\mathbf{W'} = \mathbf{W}(G')$.
	\item If $\theta = \switch{TD}_1$:
	\begin{enumerate}
		\item Let $(\{v_3, v_5\}, t_2)$, $(\{v_4, v_6\}, t_3)$ denote the edges removed by $S$. 
		\item Restart if $m_{3,5} \geq \mu$ or $m_{4,6} \geq \mu$.
		\item Restart (\textbf{b-reject}) with probability $1 -  \frac{\underline{b}_{\switch{TD}_0}(\mathbf{W'}; 2)}{b_{\switch{TD}_0}(G', v_1 \dots v_6; 2)}$.
		\item Choose a switching type $\theta_\switch{B} \sim (\Theta_\switch{B} \cup \{ \switch{I} \}, P_\switch{B})$.
		\item If $\theta = \switch{B}_{m,n}$ for some $\switch{B}_{m,n} \in \Theta_\switch{B}$:
		\begin{enumerate}
			\item Restart if $m_{3,5} \geq 1$ or $m_{4,6} \geq 1$.
			\item Pick a uniform random $\switch{B}_{m,n}$ switching $S'$ which adds an edge with node set $\{v_3, v_5\}$ and multiplicity $m$ and an edge with node set $\{v_4, v_6\}$ and multiplicity $n$ to $G'$.
			\item Restart (\textbf{f-reject}) with probability $1 - \frac{f_{\switch{B}_{m,n}}(G')}{\overline{f}_{\switch{B}_{m,n}}(\mathbf{W'})}$.
			\item Rewire the edges according to $S'$ and let $G''$ denote the resulting graph.
			\item Restart (\textbf{b-reject}) with probability $1 - \frac{\underline{b}_{\switch{B}_{m,n}}(\mathbf{W'})}{b_{\switch{B}_{m,n}}(G'', v_1 \dots v_6)}$.
			\item Set $G' \gets G''$.
		\end{enumerate}
		\item Restart (\textbf{b-reject}) with probability $1 -  \frac{\underline{b}_{\switch{TD}_1}(\mathbf{W'}; 0)\underline{b}_{\switch{TD}_1}(\mathbf{W'}; 1)}{b_{\switch{TD}_1}(G', v_1 v_2; 0)b_{\switch{TD}_1}(G', v_1 v_2 v_3 v_4; 1)}$.
	\end{enumerate}
	\item Else if $\theta = \switch{TD}_0$:
	\begin{enumerate}
		\item Let $(\{v_3, v_7\}, t_2)$, $(\{v_4, v_8\}, t_3)$, $(\{v_5, v_9\}, t_4)$, $(\{v_6, v_{10}\}, t_5)$ denote the edges removed by $S$. 
		\item Restart if $m_{3,7} \geq \mu$, $m_{4,8} \geq \mu$, $m_{5,9} \geq \mu$ or $m_{6,10} \geq \mu$.
		\item Restart (\textbf{b-reject}) with probability $1 -  \frac{\underline{b}_{\switch{TD}_0}(\mathbf{W'}; 2)}{b_{\switch{TD}_0}(G', v_1 \dots v_{10}; 2)}$.
		\item Choose a switching type $\theta_\switch{C} \sim (\Theta_\switch{C} \cup \{ \switch{I} \}, P_\switch{C})$.
		\item If $\theta = \switch{C}_{m,n,o,p}$ for some $\switch{C}_{m,n,o,p} \in \Theta_\switch{C}$:
		\begin{enumerate}
			\item Restart if $m_{3,7} \geq 1$, $m_{4,8} \geq 1$, $m_{5,9} \geq 1$ or $m_{6,10} \geq 1$.
			\item Pick a uniform random $\switch{C}_{m,n,o,p}$ switching $S'$ which adds an edge with node set $\{v_3, v_7\}$ and multiplicity $m$, an edge with node set $\{v_4, v_8\}$ and multiplicity $n$, an edge with node set $\{v_5, v_9\}$ and multiplicity $o$,, and an edge with node set $\{v_6, v_{10}\}$ and multiplicity $p$ to $G'$.
			\item Restart (\textbf{f-reject}) with probability $1 - \frac{f_{\switch{C}_{m,n,o,p}}(G')}{\overline{f}_{\switch{C}_{m,n,o,p}}(\mathbf{W'})}$.
			\item Rewire the edges according to $S'$ and let $G''$ denote the resulting graph.
			\item Restart (\textbf{b-reject}) with probability $1 - \frac{\underline{b}_{\switch{C}_{m,n,o,p}}(\mathbf{W'})}{b_{\switch{C}_{m,n,o,p}}(G'', v_1 \dots v_{10})}$.
			\item Set $G' \gets G''$.
		\end{enumerate}
		\item Restart (\textbf{b-reject}) with probability $1 -  \frac{\underline{b}_{\switch{TD}_0}(\mathbf{W'}; 0)\underline{b}_{\switch{TD}_0}(\mathbf{W'}; 1)}{b_{\switch{TD}_0}(G', v_1 v_2; 0)b_{\switch{TD}_0}(G', v_1 \dots v_6; 1)}$.
	\end{enumerate}
	\item Set $G \gets G'$.
\end{enumerate}

\noindent
It remains to define the quantities required for the f- and b-rejection steps.
For the f-rejection in step $4$, define $f_{\switch{TD}_1}(G)$ and $f_{\switch{TD}_0}(G)$ as the number of $\switch{TD}_1$ and $\switch{TD}_0$ switchings which can be performed on the graph $G$, respectively.
The corresponding upper bounds are
\begin{equation}
	\nonumber
	\hfill \overline{f}_{\switch{TD}_1}(\mathbf{W}) = M^2 T^3, \quad \quad \quad \quad \overline{f}_{\switch{TD}_0}(\mathbf{W}) = M^4 T^6. \hfill
\end{equation}

\noindent
For the b-rejections in steps $6c$ and $7c$, define $b_\switch{TD_1}(G', v_1 \dots v_6; 2)$ as the number of timestamps $t_2, t_3 \in [1, T]$ such that $(\{v_3, v_5\}, t_2)$, $(\{v_4, v_6\}, t_3)$ do not exist in $G'$ and $b_\switch{TD_0}(G', v_1 \dots v_{10}; 2)$ as the number of timestamps $t_2, t_3, t_4, t_5 \in [1, T]$ such that $(\{v_3, v_7\}, t_2)$, $(\{v_4, v_8\}, t_3)$, $(\{v_5, v_9\}, t_4)$, $(\{v_6, v_{10}\}, t_5)$ do not exist in $G'$.
The lower bounds are
\begin{equation}
	\nonumber
	\hfill \underline{b}_{\switch{TD}_1}(\mathbf{W'}; 2) = (T - (\mu - 1))^2, \quad \quad \quad \quad \underline{b}_{\switch{TD}_0}(\mathbf{W'}; 2) = (T - (\mu - 1))^4. \hfill
\end{equation}

\noindent
For the b-rejections in steps $6f$ and $7f$, define $b_\switch{TD_1}(G', v_1 \dots v_4; 1)$ as the number of simple edges $(\{v_5, v_6\}, t_6)$ in $G'$ such that $v_5$, $v_6$ are distinct from $v_1, v_2, v_3, v_4$ and $b_\switch{TD_0}(G', v_1 \dots v_6; 1)$ as the number of distinct simple edges $(\{v_7, v_8\}, t_{10})$, $(\{v_9, v_{10}\}, t_{11})$ in $G'$ such that $v_7$, $v_8$, $v_9$, $v_{10}$ are distinct from $v_1, v_2, v_3, v_4, v_5, v_6$.
Then, define $b_\switch{TD_1}(G', v_1 v_2; 0)$ as the number of simple edges $(\{v_1, v_3\}, t_{4})$, $(\{v_2, v_4\}, t_{5})$ incident at $v_1$ and $v_2$ in $G'$ and $b_\switch{TD_0}(G', v_1 v_2; 0)$ as the number of distinct simple edges $(\{v_1, v_3\}, t_{4})$, $(\{v_2, v_4\}, t_{5})$, $(\{v_1, v_5\}, t_{6})$, $(\{v_2, v_6\}, t_{7})$ incident at $v_1$ and $v_2$ in $G'$.
The lower bounds are
\begin{gather}
\begin{align}
	\nonumber
	\quad&\underline{b}_{\switch{TD}_1}(\mathbf{W'}; 1) = M - 4 B_D - 4 \Delta, & & \underline{b}_{\switch{TD}_0}(\mathbf{W'}; 1) = (M - 4 B_D - 4 \Delta)^2,&
	\\
	\nonumber
	\quad&\underline{b}_{\switch{TD}_1}(\mathbf{W'}; 0) = k_1 k_2, & & \underline{b}_{\switch{TD}_0}(\mathbf{W'}; 0) = k_1 (k_1 - 1) k_2 (k_2 - 1).&
\end{align}
\end{gather}

\noindent
For the f-rejections in steps $6eiii$ and $7eiii$, define $f_{\switch{B}_{m,n}}(G')$, $f_{\switch{C}_{m,n,o,p}}(G')$ as the number of $\switch{B}_{m,n}$, $\switch{C}_{m,n,o,p}$ switchings which can be performed on $G'$, respectively.
In addition, define
\begin{equation}
	\nonumber
	\hfill \overline{f}_{\switch{B}_{m,n}}(\mathbf{W'}) = \Delta^{2 (m+n)} T^{2 (m+n)}, \quad \quad \quad \quad \overline{f}_{\switch{C}_{m,n,o,p}}(\mathbf{W'}) = \Delta^{2 (m+n+o+p)} T^{2 (m+n+o+p)}. \hfill
\end{equation}

\noindent
For the b-rejections in steps $6ev$ and $7ev$, define $b_{\switch{B}_{m,n}}(G'', v_1 \dots v_{6})$ and $b_{\switch{C}_{m,n,o,p}}(G'', v_1 \dots v_{10})$ as the number of $\switch{B}_{m,n}$ and $\switch{C}_{m,n,o,p}$ switchings which can produce the graph $G''$, respectively.
The corresponding lower bounds are
%\begin{gather}
\begin{align}
	\nonumber
	\underline{b}_{\switch{B}_{m,n}}(\mathbf{W'}) &= (M - 4 B_D - 4 (m+n+3) \Delta)^{m+n} (T - (\Delta - 1))^{2 (m+n)},
	\\
	\nonumber
	\underline{b}_{\switch{C}_{m,n,o,p}}(\mathbf{W'}) &= (M - 4 B_D - 4 (m+n+o+p+5) \Delta)^{m+n+o+p} (T - (\Delta - 1))^{2 (m+n+o+p)}.
\end{align}
%\end{gather}

\noindent
The following lemmata ensure the correctness and efficiency of Stage 2 (proofs in \autoref{apx-subsec:uniformity-proofs} and \autoref{apx-subsec:run-time-proofs}).

\begin{lemma}
	\label{lem:st2-uniformity}
	The graph $G'$ at the end of an iteration of Stage 2 is uniformly distributed in $\mathcal{S}(\mathbf{W'})$ given that the graph $G$ at the start of the iteration is uniformly distributed in $\mathcal{S}(\mathbf{W})$.
\end{lemma}

\begin{lemma}
	\label{lem:st2-rejection}
	The probability of not restarting in Stage 2 is $\exp(- O(\Delta^3 / M T) - O(\Delta^2 / T^2))$.
\end{lemma}

\begin{lemma}
	\label{lem:st2-runtime}
	The expected run time of Stage 2 is $O(\Delta^3 / T)$.
\end{lemma}

\noindent
Once Stage 2 ends, the final graph is simple and can be output.

\section{Proof of \autoref{thm:t-gen}}
\label{subsec:thm2-proof}
It remains to show \autoref{thm:t-gen}.
We start with the uniformity of the output distribution.

\subsection{Uniformity of \textsc{T-Gen}}
\label{apx-subsec:uniformity-proofs}

To show \autoref{lem:st1-uniformity}, we first show that the upper and lower bounds specified for Stage 1 hold.
To this end, let $\mathcal{M}_0(\mathbf{D}) \subseteq \mathcal{M}(\mathbf{D})$ denote the set of temporal multigraphs which satisfy the initial conditions specified in \autoref{subsec:initial-rejection}.
Then, we obtain the following.

\begin{lemma}
	\label{lem:st1-bounds-1}
	For all $G \in \mathcal{S}(\mathbf{W}) \subseteq \mathcal{M}_0(\mathbf{D})$ we have
	\begin{align}
		\nonumber
		&f_{\switch{TL}}(G) \leq \overline{f}_{\switch{TL}}(\mathbf{W}).
	\end{align}
\end{lemma}

\begin{proof}
	By \autoref{def:tl-switching}, a $\switch{TL}$ switching involves two edges $(\{v_2, v_4\}, t_2)$, $(\{v_3, v_5\}, t_3)$ and three timestamps $t_4, t_5, t_6 \in [1, T]$.
	The total number of choices for the edges and timestamps constitutes an upper bound on the number of switchings which can be performed, and there are at most $M^2$ choices the (oriented) edges and at most $T^3$ choices the timestamps.
\end{proof}

\begin{lemma}
	\label{lem:st1-bounds-2}
	For all $G' \in \mathcal{S}(\mathbf{W'}) \subseteq \mathcal{M}_0(\mathbf{D})$ where $m_{2,4} < \mu$ and $m_{3,5} < \mu$ we have
	\begin{align}
		\nonumber
		&\underline{b}_{\switch{TL}}(\mathbf{W'}; 2) \leq b_\switch{TL}(G', v_1 v_2 v_3 v_4 v_5; 2) \leq T^2
	\end{align}
	and for all $G' \in \mathcal{S}(\mathbf{W'}) \subseteq \mathcal{M}_0(\mathbf{D})$ we have
	\begin{align}
		\nonumber
		&\underline{b}_\switch{TL}(\mathbf{W'}; 1) \leq b_\switch{TL}(G', v_1 v_2 v_3; 1) \leq M,
		\\
		\nonumber
		&\underline{b}_\switch{TL}(\mathbf{W'}; 0) = b_\switch{TL}(G', v_1; 0).
	\end{align}
\end{lemma}

\begin{proof}
	The first set of inequalities follows since there are at most $T$ and at least $T - (\mu - 1)$ available timestamps for each edge with multiplicity at most $\mu - 1$.
	For the second set of inequalities, there are at most $M$ choices for an edge, at most $2 B_L + 4 B_D$ choices such that the edge is not simple, and at most $4 \Delta$ choices such that some of the nodes are not distinct.
	For the equality, observe that $\mathbf{W'}$ (in addition to $\mathbf{D}$) determines the number of incident simple temporal edges at all nodes.
\end{proof}

\begin{lemma}
	\label{lem:st1-bounds-3}
	For all $G' \in \mathcal{S}(\mathbf{W'}) \subseteq \mathcal{M}_0(\mathbf{D})$ we have
	\begin{align}
		\nonumber
		&f_{\switch{A}_{m,n}}(G') \leq \overline{f}_{\switch{A}_{m,n}}(\mathbf{W'}).
	\end{align}
\end{lemma}

\begin{proof}
	The number of edges and timestamps needed for the $\switch{A}_{m,n}$ switching is $2 (m+n)$ and there are at most $\Delta$ choices for each incident edge, and at most $T$ choices for each timestamp.
\end{proof}

\begin{lemma}
	\label{lem:st1-bounds-4}
	For all $G'' \in \mathcal{S}(\mathbf{W'}) \subseteq \mathcal{M}_0(\mathbf{D})$ we have
	\begin{align}
		\nonumber
		&\underline{b}_{\switch{A}_{m,n}}(\mathbf{W'}) \leq b_{\switch{A}_{m,n}}(G'', v_1 \dots v_{5}) \leq M^{m+n} T^{2(m+n)}.
	\end{align}
\end{lemma}

\begin{proof}
	The number of switchings which can produce a given graph corresponds to the number of choices for edges and timestamps needed to reverse the switching.
	Reversing the $\switch{A}_{m,n}$ switching requires $m+n$ edges and $2 (m+ n)$ timestamps.
	For each edge, there are at most $M$ choices, at most $2 B_L + 4 B_D$ choices such that the edge is not simple, and at most $2 \Delta$ choices for each node already chosen such that the nodes are not distinct.
	For each timestamp, there are at most $T$ choices, and at least $T - (\Delta - 1)$ choices such that the edge does not exist in $G''$.
\end{proof}

\begin{proof}[Proof of \autoref{lem:st1-uniformity}]	Let $c_{i,j} = \sum_{1 \leq t \leq T} \mathbf{1}_{w_{i,j,t} > 0}$ denote the number of distinct temporal edges between two given nodes $v_i, v_j$ and $N(v_i)$ the (multi-)set of incident temporal edges at a given node $v_i$.
	Then by \autoref{lem:st1-bounds-1}, after the f-rejection in step $3$ a given graph $G' \in \mathcal{S}(\mathbf{W'})$ is produced with probability
	\begin{equation}
		\nonumber
		\hfill \sum_{\substack{(\{v_1, v_2\}, t_4) \neq (\{v_1, v_3\}, t_5) \in N(v_1) \\ v_1 \neq v_2, w_{1,2,4} = 1 \\ v_1 \neq v_3, w_{1,3,5} = 1}} \sum_{\substack{(\{v_4, v_5\}, t_6) \in E \\ v_4 \neq v_5, w_{4,5,6} = 1 \\ v_4 \notin \{v_1, v_2\} \\ v_5 \notin \{v_1, v_3\}}} \frac{(T - c_{2,4}) (T - c_{3,5})}{\overline{f}_\switch{TL}(\mathbf{W}) |\mathcal{S}(\mathbf{W})|} \hfill
	\end{equation}
	where $v_1$ is the node at which the temporal single-loop was removed.
	Thus, in particular, $G'$ is produced via switchings where we fix the three created edges $(\{v_1, v_2\}, t_4)$, $(\{v_1, v_3\}, t_5)$, $(\{v_4, v_5\}, t_6)$ to three given edges which satisfy the conditions with probability
	\begin{equation}
		\nonumber
		\hfill \frac{(T - c_{2,4}) (T - c_{3,5})}{{\overline{f}_\switch{TL}(\mathbf{W}) |\mathcal{S}(\mathbf{W})|}} \propto b_\switch{TL}(G', v_1 v_2 v_3 v_4 v_5; 2) \hfill
	\end{equation}
	and by \autoref{lem:st1-bounds-2} after steps $6$ and $7$ with probability proportional to $\underline{b}_\switch{TL}(\mathbf{W'}; 2) = T - (\mu - 1)$ if $m_{2,4}, m_{3,5} < \mu$ (which implies $c_{2,4}, c_{3,5} < \mu$) and probability $0$ if any of $m_{2,4} \geq \mu$, $m_{3,5} \geq \mu$.
	Now, if $m_{2,4} = m_{3,5} = 0$, we perform a type $\switch{A}_{m, n}$ switching with probability $p_\switch{A}(\switch{A}_{m, n})$, or restart the algorithm with this probability if $0 < m_{2,4}, m_{3,5} < \mu$.
	Thus, if $0 \leq m_{2,4}, m_{3,5} < \mu$, the probability of producing $G'$ via switchings which create the three fixed edges is
	\begin{equation}
		\nonumber
		\hfill p_\switch{A}(\switch{I}) \frac{\underline{b}_\switch{TL}(\mathbf{W'}; 2)}{\overline{f}_\switch{TL}(\mathbf{W}) |\mathcal{S}(\mathbf{W})|}. \hfill
	\end{equation}
	If instead $\mu \leq m_{2,4} < \min \{d_2, d_4\}$ or $\mu \leq m_{3,5} < \min \{d_3, d_5\}$, then $G'$ is produced only via an $\switch{A}_{m_{2,4}, m_{3,5}}$ switching on a graph with $m_{2,4} = m_{3,5} = 0$ and by \autoref{lem:st1-bounds-3} and \autoref{lem:st1-bounds-4}, after the f- and b-rejections in steps $9c$ and $9e$, the probability of producing $G'$ in this way is
	\begin{equation}
		\nonumber
		\hfill p_\switch{A}(\switch{A}_{m_{2,4}, m_{3,5}}) \frac{\underline{b}_{\switch{A}_{m_{2,4}, m_{3,5}}}(\mathbf{W'})}{\overline{f}_{\switch{A}_{m_{2,4}, m_{3,5}}}(\mathbf{W'})} \frac{\underline{b}_\switch{TL}(\mathbf{W'}; 2)}{\overline{f}_\switch{TL}(\mathbf{W}) |\mathcal{S}(\mathbf{W})|}. \hfill
	\end{equation}
	It is now straightforward to verify that the probabilities $p_\switch{A}(\switch{I})$ and $p_\switch{A}(\switch{A}_{m, n})$ as specified for Stage 1 equalize the expressions given above.
	Thus, after step $9$, the probability of producing a given graph $G'$ via switchings which create the three fixed edges no longer depends on $c_{2,4}$ or $c_{3,5}$.
	It only remains to show that step $10$ equalizes the probabilities over all choices of the three edges $(\{v_1, v_2\}, t_4)$, $(\{v_1, v_3\}, t_5)$ and $(\{v_4, v_5\}, t_6)$.
	To this end, observe that the probability of producing $G'$ via switchings where we only fix the choices of $(\{v_1, v_2\}, t_4)$ and $(\{v_1, v_3\}, t_5)$ is
	\begin{equation}
		\nonumber
		\hfill \sum_{\substack{(\{v_4, v_5\}, t_6) \in E \\ v_4 \neq v_5, w_{4,5,6} = 1 \\ v_4 \notin \{v_1, v_2\} \\ v_5 \notin \{v_1, v_3\}}} p_\switch{A}(\switch{I}) \frac{\underline{b}_\switch{TL}(\mathbf{W'}; 2)}{{\overline{f}_\switch{TL}(\mathbf{W}) |\mathcal{S}(\mathbf{W})|}} \propto b_\switch{TL}(G', v_1 v_2 v_3; 1) \hfill
	\end{equation}
	and by \autoref{lem:st1-bounds-2} after the second b-rejection in step $10$, this probability is proportional to $\underline{b}_\switch{TL}(\mathbf{W'}; 1)$.
	Finally, to show that the probability is equal over all choices of $(\{v_1, v_2\}, t_4)$, $(\{v_1, v_3\}, t_5)$, observe that by \autoref{lem:st1-bounds-2}, we have
	\begin{equation}
		\nonumber
		\hfill \sum_{\substack{(\{v_1, v_2\}, t_4) \neq (\{v_1, v_3\}, t_5) \in N(v_1) \\ v_1 \neq v_2, w_{1,2,4} = 1 \\ v_1 \neq v_3, w_{1,3,5} = 1}} p_\switch{A}(\switch{I}) \frac{\underline{b}_\switch{TL}(\mathbf{W'}; 1) \underline{b}_\switch{TL}(\mathbf{W'}; 2)}{{\overline{f}_\switch{TL}(\mathbf{W}) |\mathcal{S}(\mathbf{W})|}} \propto b_\switch{TL}(G', v_1; 0) = \underline{b}_\switch{TL}(\mathbf{W'}; 0) \hfill
	\end{equation}
	and thus $G'$ is produced with probability
	\begin{equation}
		\nonumber
		\hfill p_\switch{A}(\switch{I}) \frac{\underline{b}_\switch{TL}(\mathbf{W'}; 0) \underline{b}_\switch{TL}(\mathbf{W'}; 1) \underline{b}_\switch{TL}(\mathbf{W'}; 2)}{{\overline{f}_\switch{TL}(\mathbf{W}) |\mathcal{S}(\mathbf{W})|}}\hfill
	\end{equation}
	which only depends on $\mathbf{W}$, $\mathbf{W'}$ and $\mathbf{D}$.
\end{proof}

\noindent
The proof of \autoref{lem:st2-uniformity} requires showing that the upper and lower bounds for Stage 2 are correct.
To this end, let $\mathcal{M}_1(\mathbf{D}) \subseteq \mathcal{M}_0(\mathbf{D})$ denote set of temporal multigraphs which satisfy the initial conditions and which are output by Stage 1, i.e. which contain no temporal single-loops.
Then, we obtain the following results.

\begin{lemma}
	\label{lem:st2-bound-1}
	For all $G \in \mathcal{S}(\mathbf{W}) \subseteq \mathcal{M}_1(\mathbf{D})$ we have
	\begin{align}
		\nonumber
		f_{\switch{TD}_1}(G) \leq \overline{f}_{\switch{TD}_1}(\mathbf{W}),
		\\
		\nonumber
		f_{\switch{TD}_0}(G) \leq \overline{f}_{\switch{TD}_0}(\mathbf{W}).
	\end{align}
\end{lemma}

\begin{proof}
	There are at most $M^2$ choices for the (oriented) edges $(\{v_3, v_5\}, t_2)$, $(\{v_4, v_6\}, t_3)$ and at most $T^3$ choices for the timestamps $t_4, t_5, t_6 \in [1, T]$ needed to perform a $\switch{TD}_1$ switching.
	The $\switch{TD}_0$ switching uses twice as many edges and timestamps.
\end{proof}

\begin{lemma}
	\label{lem:st2-bound-2}
	For all $G' \in \mathcal{S}(\mathbf{W'}) \subseteq \mathcal{M}_1(\mathbf{D})$ where $m_{3,5} < \mu$ and $m_{4,6} < \mu$ we have
	\begin{align}
		\nonumber
		\underline{b}_{\switch{TD}_1}(\mathbf{W'}; 2) \leq b_\switch{TD_1}(G', v_1 \dots v_6; 2) \leq T^2,
	\end{align}
	for all $G' \in \mathcal{S}(\mathbf{W'}) \subseteq \mathcal{M}_1(\mathbf{D})$ where $m_{3,7}, m_{4,8}, m_{5,9}, m_{6,10} < \mu$ we have
	\begin{align}
		\nonumber
		\underline{b}_{\switch{TD}_0}(\mathbf{W'}; 2) \leq b_\switch{TD_0}(G', v_1 \dots v_{10}; 2) \leq T^4
	\end{align}
	and for all $G' \in \mathcal{S}(\mathbf{W'}) \subseteq \mathcal{M}_1(\mathbf{D})$ we have
	\begin{align}
		\nonumber
		&\underline{b}_{\switch{TD}_1}(\mathbf{W'}; 1) \leq b_\switch{TD_1}(G', v_1 \dots v_4; 1) \leq M,
		\\
		\nonumber
		&\underline{b}_{\switch{TD}_0}(\mathbf{W'}; 1) \leq b_\switch{TD_0}(G', v_1 \dots v_6; 1) \leq M^2,
		\\
		\nonumber
		&\underline{b}_{\switch{TD}_1}(\mathbf{W'}; 0) = b_\switch{TD_1}(G', v_1 v_2; 0),
		\\
		\nonumber
		&\underline{b}_{\switch{TD}_0}(\mathbf{W'}; 0) = b_\switch{TD_0}(G', v_1 v_2; 0).
	\end{align}
\end{lemma}

\begin{proof}
	The first two sets of inequalities follow since there are at most $T$ and at least $T - (\mu - 1)$ available timestamps for each edge.
	For the third and fourth sets of inequalities, observe that there are at most $M$ choices for an edge, at most $4 B_D$ choices such that the edge is not simple, and at most $4 \Delta$ choices such that the nodes are not distinct.
	The equalities follow from the observation that $\mathbf{W'}$ (in addition to $\mathbf{D}$) determines the number of incident simple temporal edges at each node.
\end{proof}

\begin{lemma}
	\label{lem:st2-bound-3}
	For all $G' \in \mathcal{S}(\mathbf{W'}) \subseteq \mathcal{M}_1(\mathbf{D})$ we have
	\begin{align}
		\nonumber
		&f_{\switch{B}_{m,n}}(G') \leq \overline{f}_{\switch{B}_{m,n}}(\mathbf{W'}),
		\\
		\nonumber
		&f_{\switch{C}_{m,n,o,p}}(G') \leq \overline{f}_{\switch{C}_{m,n,o,p}}(\mathbf{W'}).
	\end{align}
\end{lemma}

\begin{proof}
	The number of edges and timestamps needed for the $\switch{B}_{m,n}$ switching is $2 (m+n)$, the number of edges and timestamps needed for the $\switch{C}_{m,n,o,p}$ switching is $2 (m+n+o+p)$, and there are at most $\Delta$ choices for each incident edge, and at most $T$ choices for each timestamp.
\end{proof}

\begin{lemma}
	\label{lem:st2-bound-4}
	For all $G'' \in \mathcal{S}(\mathbf{W'}) \subseteq \mathcal{M}_1(\mathbf{D})$ we have
	\begin{align}
		\nonumber
		&\underline{b}_{\switch{B}_{m,n}}(\mathbf{W'}) \leq b_{\switch{B}_{m,n}}(G'', v_1 \dots v_{6}) \leq M^{m+n} T^{2 (m+n)},
		\\
		\nonumber
		&\underline{b}_{\switch{C}_{m,n,o,p}}(\mathbf{W'}) \leq b_{\switch{C}_{m,n,o,p}}(G'', v_1 \dots v_{10}) \leq M^{m+n+o+p} T^{2 (m+n+o+p)}.
	\end{align}
\end{lemma}

\begin{proof}
	The number of edges and timestamps needed to reverse the $\switch{B}_{m,n}$ switching is $m+n$ and $2 (m+n)$, the number of edges and timestamps needed to reverse the $\switch{C}_{m,n,o,p}$ switching is $m+n+o+p$ and $2 (m+n+o+p)$.
	There are at most $M$ choices for each edge, at most $4 B_D$ choices such that the edge is not simple, and at most $2 \Delta$ choices for each node already chosen such that the nodes are not distinct.
	Finally, there are at most $T$ and at least $T - (\Delta - 1)$ choices for each timestamp such that the edge does not exist in $G''$.
\end{proof}

\begin{proof}[Proof of \autoref{lem:st2-uniformity}]
	Using \autoref{lem:st2-bound-1}, \autoref{lem:st2-bound-2}, \autoref{lem:st2-bound-3}, and \autoref{lem:st2-bound-4} in a similar style argument as in the proof of \autoref{lem:st1-uniformity}, after removing a temporal double-edge between two nodes $v_1$ and $v_2$ in step $5$, a given graph $G' \in \mathcal{S}(\mathbf{W'})$ is produced with probability
	\begin{equation}
		\nonumber
		\hfill p(\switch{TD}_1) p_\switch{B}(\switch{I}) \frac{\underline{b}_{\switch{TD}_1}(\mathbf{W'}; 0) \underline{b}_{\switch{TD}_1}(\mathbf{W'}; 1) \underline{b}_{\switch{TD}_1}(\mathbf{W'}; 2)}{\overline{f}_{\switch{TD}_1}(\mathbf{W}) |\mathcal{S}(\mathbf{W})|} \hfill
	\end{equation}
	if $m_{1,2}(G') = 1$ and with probability 
	\begin{equation}
		\nonumber
		\hfill p(\switch{TD}_0) p_\switch{C}(\switch{I}) \frac{\underline{b}_{\switch{TD}_0}(\mathbf{W'}; 0) \underline{b}_{\switch{TD}_0}(\mathbf{W'}; 1) \underline{b}_{\switch{TD}_0}(\mathbf{W'}; 2)}{\overline{f}_{\switch{TD}_0}(\mathbf{W}) |\mathcal{S}(\mathbf{W})|} \hfill
	\end{equation}
	if $m_{1,2}(G') = 0$.
	Specifying the probabilities $p(\switch{TD}_1)$ and $p(\switch{TD}_0)$ as done for Stage 2 then suffices to equalize the probabilities over all graphs in $\mathcal{S}(\mathbf{W'})$.
\end{proof}

\noindent
We are now able to show the following.

\begin{lemma}
	\label{lem:t-gen-uniformity}
	Given a tuple $\mathbf{D} = (\mathbf{d}, T)$ as input, \textsc{T-Gen} outputs a uniform random sample $G \in \mathcal{G}(\mathbf{D})$.
\end{lemma}

\begin{proof}
	If the initial graph $G$ is simple, then the claim follows by \autoref{thm:tcm-output-distribution}.
	Otherwise the initial graph $G$ is uniformly distributed in the set $\mathcal{S}(\mathbf{W}(G)) \subseteq \mathcal{M}(\mathbf{D})$ for some $\mathbf{W}(G) \neq \mathbf{0}^{n\times n \times T}$.
	If $G$ satisfies the initial conditions, then all entries $\mathbf{W}(G)_{i,i,t}$ are either $0$ or $1$ and all entries $\mathbf{W}(G)_{i,j,t}$ such that $i \neq j$ are either $0$ or $2$.
	Now, it is straightforward to check that each iteration of Stage 1 corresponds to a map $\mathcal{S}(\mathbf{W}) \to \mathcal{S}(\mathbf{W'})$ where $\mathbf{W'}$ is the tensor obtained from $\mathbf{W} = \mathbf{W}(G)$ by setting exactly one entry $\mathbf{W}_{i,i,t} = 1$ to $0$, and Stage 1 ends once all such entries have been set to $0$.
	Similarly, each iteration of Stage 2 corresponds to a map $\mathcal{S}(\mathbf{W}) \to \mathcal{S}(\mathbf{W'})$ where $\mathbf{W'}$ is the tensor obtained from $\mathbf{W}$ by setting exactly two entries $\mathbf{W}_{i,j, t} = \mathbf{W}_{j, i, t} = 2$ where $i \neq j$ to $0$, and Stage 2 ends once all such entries have been set to $0$.
	Thus, after Stage 1 and Stage 2 end, the final graph $G$ is a simple temporal graph with $\mathbf{W}(G) = \mathbf{0}^{n\times n \times T}$, and by \autoref{lem:st1-uniformity} and \autoref{lem:st2-uniformity}, this graph is uniformly distributed in $\mathcal{S}(\mathbf{0}^{n\times n \times T}) = \mathcal{G}(\mathbf{D})$ as claimed.
\end{proof}

\noindent
We move on to the run time proof.

\subsection{Runtime of \textsc{T-Gen}}
\label{apx-subsec:run-time-proofs}

%\begin{proof}[Proof of \autoref{lem:initial-restart}]
%	By \autoref{lem:tcm-temporal-loops-multi-edges} we have $\mathbb E[L] = B_L$ and $\mathbb E[D / 2] = B_D$ so the first two conditions are satisfied if each of these random variables is at most its expected value.
%	As was noted in \cite{DBLP:journals/jal/McKayW90} the numbers of non-simple edges in a random pairing output by the classical configuration model are asymptotically independent poisson random variables, and the same result extends to the temporal configuration model.
%	In addition, by \autoref{lem:tcm-kappa-lambda-eta}, the conditions of no temporal double-loops, no temporal triple-edges, and no more than $\lambda$ and $\kappa$ incident temporal single-loops and incident temporal double-edges are all satisfied with probability $1 - o(1)$.
%	The last condition of at least two simple edges at each node incident with a temporal non-simple edge clearly holds for any node with degree $d > \lambda + \kappa + 1$.
%	Furthermore, it is straightforward to check that the expected number of degree $d \leq \lambda + \kappa + 1 = O(1)$ nodes incident with at least one temporal single-loop or temporal double-edge is at most $O(n / M)$ and by the same argument as given for the first condition there is at least a constant probability that no such node is incident with a temporal single-loop or temporal double-edge.
%\end{proof}

%\noindent
The following additional results are needed for the proofs of \autoref{lem:edge-bound-probability}, \autoref{lem:st1-rejection}, \autoref{lem:st2-rejection}, \autoref{lem:st1-runtime}, and \autoref{lem:st2-runtime}.

\begin{lemma}
	\label{lem:st1-auxiliary}
	Given that $\Delta^{2+\epsilon} = O(M)$ for a constant $\epsilon > 0$ and $T - \Delta = \Omega(T)$, we have $p_\switch{A}(\switch{I}) = 1 - o(\Delta^{-1})$.
\end{lemma}

\begin{proof}
	Let $k = m + n$.
	Then, the probability of choosing a type $\mathsf{A}_{m,n}$ switching is at most
	\begin{equation}
		\nonumber
		\hfill p(\mathsf{A}_{m,n}) = p_\switch{A}(\switch{I}) \frac{\overline{f}_{\mathsf{A}_{m,n}}(\mathbf{W'})}{\underline{b}_{\mathsf{A}_{m,n}}(\mathbf{W'})} < \frac{\Delta^{2 k} T^{2 k}}{(M - B_L - B_D - 4 (k + 3) \Delta)^{k} (T - (\Delta - 1))^{2 k}}. \hfill
	\end{equation}
	Now, if $\Delta^{2+\epsilon} = O(M)$ and $T - \Delta = \Omega(T)$, then 
	\begin{equation}
		\nonumber
		\hfill \frac{\Delta^{2 k} T^{2 k}}{(M - B_L - B_D - 4 (k + 3) \Delta)^{k} (T - (\Delta - 1))^{2 k}} = O\left(\Delta^{-\epsilon k}\right) = o(\Delta^{-k / \mu}) \hfill
	\end{equation}
	by $B_L + B_D + 4 (k + 3) \Delta = O(\Delta^2) = o(M)$ and $\epsilon > \frac{1}{\mu}$.
	Thus, the type $\switch{I}$ switching is chosen with probability at least
	\begin{equation}
		\nonumber
		\hfill p_\switch{A}(\switch{I}) = 1 - \sum_{\substack{0 \leq m,n < \Delta \\ \mu \leq \max \{m,n\}}} p_\switch{A}(\mathsf{A}_{m,n}) > 1 -  \sum_{\mu \leq k < 2 \Delta} \binom{k+1}{1} o\left(\Delta^{-k / \mu}\right) = 1 - o\left(\Delta^{-1}\right)\hfill
	\end{equation}
	as claimed.
\end{proof}

\begin{lemma}
	\label{lem:st2-auxiliary}
	Given that $\Delta^{2+\epsilon} = O(M)$ for a constant $\epsilon > 0$ and $T - \Delta = \Omega(T)$, we have $p_\switch{B}(\switch{I}) = 1 - o(\Delta^{-1})$ and $p_\switch{C}(\switch{I}) = 1 - o(\Delta^{-1})$.
\end{lemma}

\begin{proof}
	By a similar argument as in the proof of \autoref{lem:st1-auxiliary}.
\end{proof}

\begin{proof}[Proof of \autoref{lem:edge-bound-probability}]
	By \autoref{lem:tcm-kappa-lambda-eta}, the highest multiplicity of any ordinary multi-edge in the initial graph is at most $\eta = \lfloor 2 + 2 / \epsilon \rfloor$ with high probability.
	To complete the proof, we show that with high probability no two edges created due to switchings share the same node set, which in turn implies that the highest multiplicity of any ordinary multi-edge is $\eta + 1 = \lfloor 3 + 2 / \epsilon \rfloor := \mu$.
	
	First, observe that \autoref{lem:st1-auxiliary} and \autoref{lem:st2-auxiliary} imply that the probability of performing at least one auxiliary switching in either of Stage 1 or Stage 2 is at most $(B_L + B_D) o(\Delta^{-1}) = o(1)$.
	With the remaining probability, only $\switch{TL}$, $\switch{TD}_1$, and $\switch{TD}_0$ switchings are performed.
	The $\switch{TL}$ switching creates three edges with node sets $\{v_1, v_2\}$, $\{v_1, v_3\}$ and $\{v_4, v_5\}$ where $v_1$ is incident with a temporal single-loop, and where $v_2$ and $v_4$, $v_3$ and $v_5$ are determined by choosing the switching uniformly at random which implies that these nodes are incident with edges $(\{v_2, v_4\}, t_2)$, $(\{v_3, v_5\}, t_3)$ chosen uniformly at random among all edges which satisfy the conditions.
	Likewise, the $\switch{TD}_1$ switching creates three edges with node sets $\{v_1, v_3\}$, $\{v_2, v_4\}$ and $\{v_5, v_6\}$ where $v_1$ and $v_2$ are incident with a temporal double-edge, $v_3$ and $v_5$, $v_4$ and $v_6$ are incident with edges satisfying the conditions chosen uniformly at random, and the $\switch{TD}_0$ switching adds two such sets of edges.
	Now, observe that when performing a $\switch{TL}$ switching at some node $v_1$, there are at least $M - 2 B_L - 4 B_D - 2 \Delta$ choices for the edge $(\{v_2, v_4\}, t_2)$ and at least $M - 2 B_L - 4 B_D - 3 \Delta$ choices for the second edge $(\{v_3, v_5\}, t_3)$.
	Thus, the probability that a given node $v$ with degree $d$ takes on the role of one of $v_2$, $v_3$, $v_4$, $v_5$ in such a switching chosen uniformly at random is at most
	\begin{equation}
		\nonumber
		\hfill \frac{2d}{M - 2 B_L - 4 B_D - 2 \Delta} + \frac{2d}{M - 2 B_L - 4 B_D - 3 \Delta} = O\left(\frac{\Delta}{M}\right) \hfill
	\end{equation}
	due to $d \leq \Delta$ and $O(B_L + B_D + \Delta)= o(M) \implies M - O(B_L - B_D - \Delta) = \Omega(M)$.
	Similar calculations give the same asymptotic probability of a given node being involved in a $\switch{TD}_1$ or $\switch{TD}_0$ switching.
	
	This leads to the following bounds on the probability of creating two edges with the same node set in terms of the number of iterations of a stage performed.
	Fix any stage, and any node $v$, and let $i$ denote the number of iterations performed overall, and $j$ the number of iterations performed such that the non-simple edge removed is incident at $v$.
	Then, the number of edges created overall is $O(i)$ and the expected number of edges created which are incident at $v$ is $O(j + i \Delta / M)$.
	Thus, the probability that the next switching performed at $v$ creates an edge which shares the same node set as any edges created prior is $O(i \Delta^2 / M^2)$ for edges not incident at $v$, and $O(j \Delta / M + i \Delta^2 / M^2)$ for edges incident at $v$.
	
	Now, recall that by the initial conditions, there are at most $\lambda = O(1)$ incident temporal single-loops and at most $\kappa = O(1)$ incident temporal double-edges at any node, and at most $B_L = M_2 / M = O(\Delta)$ temporal single-loops and at most $B_D = M_2^2 / M T= O(\Delta^2 / T)$ temporal double-edges overall.
	Then, starting with Stage 1, the probability of creating two edges with share the same node set due to any of the $\switch{TL}$ switchings is at most
	\begin{align}
		\nonumber
		\sum_{1 \leq i \leq B_L} \sum_{1 \leq j \leq \lambda} \left(O\left(i \frac{\Delta^2}{M^2}\right) + O\left(j \frac{\Delta}{M} + i \frac{\Delta^2}{M^2}\right)\right) &< B_L \lambda \left(O\left(B_L \frac{\Delta^2}{M^2}\right) + O\left(\lambda \frac{\Delta}{M}\right)\right)
		\\
		\nonumber
		&= O\left(\frac{\Delta^4}{M^2} + \frac{\Delta^2}{M}\right)
		\\
		\nonumber
		&= o(1)
	\end{align}
	by $B_L = O(\Delta)$, $\lambda = O(1)$, and $\Delta^{2+\epsilon} = O(M)$.
	Likewise, for Stage 2, the probability of creating two edges with share the same node set due to any of the $\switch{TD}_1$ or $\switch{TD}_0$ switchings is at most
	\begin{gather}
	\begin{align}
		\nonumber
		\quad \quad \sum_{1 \leq i \leq B_D} &\sum_{1 \leq j \leq \kappa} \left(O\left(\left(B_L + i\right) \frac{\Delta^2}{M^2}\right) + O\left(\left(\lambda + j \right) \frac{\Delta}{M} + \left(B_L + i\right) \frac{\Delta^2}{M^2}\right)\right)
		\\
		\nonumber
		&< B_D \kappa \left(O\left(\left(B_L + B_D\right) \frac{\Delta^2}{M^2}\right) + O\left(\left(\lambda + \kappa \right) \frac{\Delta}{M}\right)\right)
		\\
		\nonumber
		&= O\left(\frac{\Delta^6}{T^2 M^2} + \frac{\Delta^5}{T M^2} + \frac{\Delta^3}{T M}\right) 
		\\
		\nonumber
		&= o(1)
	\end{align}
	\end{gather}
	by $B_D = O(\Delta^2 / T)$, $\kappa = O(1)$, $\Delta = O(T)$ and $\Delta^{2+\epsilon} = O(M)$.
\end{proof}

\begin{proof}[Proof of \autoref{lem:st1-rejection}]
	By \autoref{lem:st1-auxiliary} the probability that Stage 1 restarts in steps $9a-f$ is smaller than $B_L o(\Delta^{-1}) = o(1)$ and by \autoref{lem:edge-bound-probability} we can assume that the highest multiplicity of an edge in $G$ is at most $\mu$.
	Hence, the probability of not restarting is at most the probability of not f- or b-rejecting in steps $3$, $7$ or $10$ under this assumption.

	The probability of an f-rejection in step $3$ equals the probability of choosing edges $(\{v_2, v_4\}, t_2)$, $(\{v_3, v_5\}, t_3)$ and timestamps $t_4, t_5, t_6 \in [1, T]$ which do not fulfill the conditions defined for the $\switch{TL}$ switching.
	For each edge, this probability is at most $(2 B_L + 4 B_D + 3 \Delta) / M$ as there are at least $M$ choices for each (oriented) edge, at most $2 B_L$ choices for a loop, at most $4 B_D$ choices for an edge contained in a temporal double-edge and at most $3 \Delta$ choices for an edge such that $v_2, v_3, v_4, v_5$ is not distinct from $v_1$ or $v_4$ is not distinct from $v_5$.
	For the timestamps, we can assume that $m_{1,2}, m_{1,3}, m_{4,5} \leq \mu$ so for each timestamp the probability of a rejection is at most $(\mu - 1) / T$.
	Then, the probability of not f-rejecting in a given iteration is at least
	\begin{equation}
		\nonumber
		\hfill \frac{(M - 2 B_L - 4 B_D - 3 \Delta)^2 (T - (\mu - 1))^3}{M^2 T^3} = \left(1 - O\left(\frac{\Delta}{M}\right)\right)^2 \left(1 - O\left(\frac{\mu}{T}\right)\right)^3. \hfill
	\end{equation}
	In addition, by \autoref{lem:st1-bounds-2}, the probability of a b-rejection in step $7$ is at most $(\mu - 1) / T$ for each timestamp, and the probability of a b-rejection in step $10$ is at most $(2 B_L + 4 B_D + 4 \Delta) / M$.
	Thus, the probability of not b-rejecting in a given iteration is at least
	\begin{equation}
		\nonumber
		\hfill \frac{(M - 2 B_L - 4 B_D - 4 \Delta) (T - (\mu - 1))^2}{M T^2} = \left(1 - O\left(\frac{\Delta}{M}\right)\right) \left(1 - O\left(\frac{\mu}{T}\right)\right)^2. \hfill
	\end{equation}
	Finally, by the initial conditions there are at most $B_L = M_2 / M < \Delta$ iterations of Stage 1, and by using $\mu = O(1)$, the probability that the algorithm does not restart in Stage 1 is at least
	\begin{equation}
		\nonumber
		\hfill \left(\left(1 - O\left(\frac{\Delta}{M}\right)\right)^2 \left(1 - O\left(\frac{\mu}{T}\right)\right)^3\right)^{B_L} = \exp \left(- O\left(\frac{\Delta^2}{M}\right) - O\left(\frac{\Delta}{T}\right) \right). \hfill \qedhere \qed
	\end{equation}
\end{proof}

\begin{proof}[Proof of \autoref{lem:st1-runtime}]
	By the initial conditions, there are at most $B_L = M_2 / M < \Delta$ iterations of Stage 1.
	Thus, the claim follows if an iteration runs in expected time $O(\Delta)$.
	The steps of an iteration which require attention are the f-rejection in step $3$, the b-rejections in step $7$ and $10$, choosing an auxiliary switching in steps $8-9$, the f-rejection in step $9c$, and the b-rejection in step $9f$.
	
	%In the following, we use standard assumptions by which operations such as flipping a biased coin and sampling an integer from a given interval can be implemented in time $O(1)$.
	%We also assume that the implementation maintains lists of the temporal single-loops in the graph, the incident edges at each node, and the multiplicities of each edge.
	%Clearly, building up these data structures for the initial graph is possible in time $O(M)$ and updating them after performing a switching which rewires $k$ edges is possible in time $O(k)$.
	
	First, note that as observed by \cite{DBLP:journals/jal/McKayW90}, there is a simple trick to implement an f-rejection step at little additional cost.
	In the case of step $3$, it suffices to choose the two edges and three timestamps required for the $\switch{TL}$ switching uniformly at random and restart if those choices do not satisfy the conditions given in \autoref{def:tl-switching}.
	Then, since there are $\overline{f}_\switch{TL}(\mathbf{W}) = M^2 T^3$ ways to choose two (oriented) edges and three timestamps and $f_\switch{TL}(G)$ choices yield a switching which can be performed on the current graph $G$, we restart with the desired probability of $f_\switch{TL}(G) / \overline{f}_\switch{TL}(\mathbf{W})$.
	
	The b-rejections in step $7$ and $10$ require computing the quantities $b_\switch{TL}(G', v_1 v_2 v_3 v_4 v_5; 2)$, and $b_\switch{TL}(G', v_1 v_2 v_3; 1)$, $b_\switch{TL}(G', v_1; 0)$.
	Computing $b_\switch{TL}(G', v_1 v_2 v_3 v_4 v_5; 2)$ only requires look ups of the multiplicities of $\{v_2, v_4\}$ and $\{v_3, v_5\}$ which  take time $O(1)$ if the implementation maintains a data structure to store the multiplicity of the edges in the graph.
	To compute $b_\switch{TL}(G', v_1 v_2 v_3; 1)$, it suffices to iterate trough the lists of incident edges at nodes $v_1$, $v_2$, $v_3$ and subtract the number of simple edges which collide with those nodes from the total number of simple edges, which takes time $O(\Delta)$.
	Additionally, computing $b_\switch{TL}(G', v_1; 0)$ can be implemented in time $O(1)$ by maintaining the number of simple edges incident at each node.
	
	Computing the type distribution for steps $8-9$ naively would take time $\Theta(\Delta^2)$, but there is a simple trick to speed-up the computation.
	Re-purposing the proof of \autoref{lem:st1-auxiliary}, we see that
	\begin{equation}
		\nonumber
		\hfill \sum_{\substack{0 \leq m,n < \Delta \\ \mu \leq \max\{m, n\} \\ m+n = k}} p_\switch{A}(\switch{A}_{m,n}) < \binom{k+1}{1} p_{\switch{A}}(\switch{A}_{k,0}) < \binom{k+1}{1} \frac{\overline{f}_{\switch{A}_{k,0}}(\mathbf{W})}{\underline{b}_{\switch{A}_{k,0}}(\mathbf{W'})} =: \overline{p}_\switch{A}(\switch{A}_{k}). \hfill
	\end{equation}
	Thus, in time $O(\Delta)$, we can compute
	\begin{equation}
		\nonumber
		\hfill \underline{p}_\switch{A}(\switch{I}) := 1 - \sum_{\mu \leq k < 2 \Delta - 1} \overline{p}_\switch{A}(\switch{A}_{k}) \hfill
	\end{equation}
	as a lower bound on the probability of performing the identity switching, and with this probability step $9$ can be skipped.
	Otherwise, if one of the auxiliary switchings $\switch{A}_{m,n}$ where $m + n = k$ is chosen, then we compute the exact probabilities of these at most $k + 1 = O(\Delta)$ switchings, and either choose a switching in accordance with the exact probabilities or restart the algorithm with probability
	\begin{equation}
		\nonumber
		\hfill 1 - \sum_{\substack{0 \leq m,n < \Delta \\ \mu \leq \max\{m, n\} \\ m+n = k}} \frac{p_\switch{A}(\switch{A}_{m,n})}{\overline{p}_\switch{A}(\switch{A}_{k})} \hfill
	\end{equation}
	to correct for the overestimate.
	
	Finally, by \autoref{lem:st1-auxiliary} the probability that steps $9a-f$ are executed is at most $o(\Delta^{-1})$, so it suffices if these steps can be implemented in time $O(\Delta^2)$.
	For the f-rejection in step $9c$, we first pick the incident edges of $v_2$ and $v_4$ and the timestamps uniformly at random, and then restart if these choices do not satisfy the conditions in step \autoref{def:amn-switching}.
	This results in a restart probability of $f_{\switch{A}_{m,n}}(G') / (d_2^m d_4^n T^{2 (m+n)})$, and to reach the desired probability of $f_{\switch{A}_{m,n}}(G') / \overline{f}_{\switch{A}_{m,n}}(\mathbf{W'})$, it suffices to restart with probability $1 - (d_2^m d_4^n T^{2 (m+n)}) / (\Delta T)^{2 (m+n)}$ afterwards.
	To implement the b-rejection in step $9f$, we need to compute $b_{\switch{A}_{m,n}}(G'', v_1 v_2 v_3 v_4 v_5)$ which is the number of $\switch{A}_{m,n}$ switchings which can produce the graph $G''$.
	As this is equal to the number of ways to reverse an $\switch{A}_{m,n}$ switching on the graph $G''$, we can compute this quantity as the number of choices for $m$ and $n$ edges between $v_2, v_4$ and $v_3, v_5$, respectively, $m + n$ additional edges, and $2 (m+n)$ timestamps which satisfy the conditions.
	Computing the number of choices for the edges between $v_2, v_4$ and $v_3, v_5$ can be done in time $O(1)$ by looking up the number of simple edges between these nodes.
	Choosing an additional edge, and checking if it satisfies the conditions, e.g. if it is simple, and does not share nodes with the other edges, takes time $O(1)$ for the first edges chosen and time $O(\Delta)$ for the last edges, so for all $m + n < 2 \Delta - 1$ edges this takes at most time $O(\Delta^2)$.
	Finally, computing the number of available timestamps again only requires looking up the multiplicities of $2 (m+n) = O(\Delta)$ edges and takes time $O(\Delta)$.
\end{proof}

\begin{proof}[Proof of \autoref{lem:st2-rejection}]
	By \autoref{lem:st2-auxiliary} the probability that Stage 2 restarts in steps $6ei-vi$ or $7ei-vi$ is smaller than $B_D o(\Delta^{-1}) = o(1)$ and by \autoref{lem:edge-bound-probability} we can assume that the highest multiplicity of an edge in $G$ is at most $\mu$.
	Hence, the probability of not restarting is at most the probability of not f- or b-rejecting in steps $4$, $6c$, $6f$, $7c$ or $7f$ under this assumption.
	Furthermore, it is straightforward to check that the probability of f- or b-rejecting in a given iteration is larger if $\theta = \switch{TD}_0$ so we focus on this case.
	
	Using a similar argument as for \autoref{lem:st1-rejection}, the probability of not f-rejecting in step $3$ of a given iteration is at least
	\begin{equation}
		\nonumber
		\hfill \frac{(M - 4 B_D - 5 \Delta)^4 (T - (\mu - 1))^6}{M^4 T^6} = \left(1 - O\left(\frac{\Delta}{M}\right)\right)^4 \left(1 - O\left(\frac{\mu}{T}\right)\right)^6. \hfill
	\end{equation}
	In addition, by \autoref{lem:st2-bound-2} and \autoref{lem:st2-bound-3}, the probability of a b-rejection in step $6c$ or $7c$ is at most $(\mu - 1) / T$ for each timestamp, and the probability of a b-rejection in step $6f$ or $7f$ is at most $(4 B_D + 4 \Delta) / M$.
	Then, the probability of not b-rejecting in a given iteration is at least
	\begin{equation}
		\nonumber
		\hfill \frac{(M - 4 B_D - 4 \Delta)^2 (T - (\mu - 1))^4}{M^2 T^4} = \left(1 - O\left(\frac{\Delta}{M}\right)\right)^2 \left(1 - O\left(\frac{\mu}{T}\right)\right)^4. \hfill
	\end{equation}
	Finally, by the initial conditions there are at most $B_D = M_2^2 / M T < \Delta^2 / T$ iterations of Stage 2, and thus the probability that the algorithm does not restart in Stage 2 is at least
	\begin{equation}
		\nonumber
		\hfill \left(\left(1 - O\left(\frac{\Delta}{M}\right)\right)^4 \left(1 - O\left(\frac{\mu}{T}\right)\right)^6\right)^{B_D} = \exp \left(- O\left(\frac{\Delta^3}{M T}\right) - O\left(\frac{\Delta^2}{T^2}\right) \right). \hfill \qedhere \qed
	\end{equation}
\end{proof}

\begin{proof}[Proof of \autoref{lem:st2-runtime}]
	By the initial conditions, there are at most $B_D = M_2^2 / M T < \Delta^2 / T$ iterations of Stage 2, so the claim follows if an iteration runs in expected time $O(\Delta)$.
	We focus on the f-rejection in step $4$, the b-rejections in steps $6c$, $6f$ and $7c$, $7f$, choosing an auxiliary switching in steps $6d-e$ and $7d-e$, the f-rejections in step $6eiii$ and $7eiii$, and the b-rejections in step $6ev$ and $7ev$.
	
	To implement the f-rejection in step $4$, it again suffices to choose the two (or four) edges and three (or six) timestamps needed for the $\switch{TD}_1$ or $\switch{TD}_0$ switching uniformly at random and restart if those choices do not satisfy the conditions given in \autoref{def:td1-switching}.
	By similar arguments as in the proof of \autoref{lem:st1-runtime}, the quantities needed for the b-rejections in step $6c$, $6f$ and $7c$, $7f$ can be computed in time $O(\Delta)$.
	
	To choose an auxiliary switching in steps $6d-e$ the method described in the proof of \autoref{lem:st1-runtime} can be re-used.
	For steps $7d-e$, define
	\begin{equation}
		\nonumber
		\hfill \sum_{\substack{0 \leq m,n,o,p \leq \Delta \\ \mu < \max\{m, n,o,p\} \\ m+n+o+p = k}} p_\switch{C}(\switch{C}_{m,n,o,p}) < \binom{k+1}{3} \frac{\overline{f}_{\switch{C}_{k,0,0,0}}(\mathbf{W})}{\underline{b}_{\switch{C}_{k,0,0,0}}(\mathbf{W'})} =: \overline{p}_\switch{C}(\switch{C}_{k}). \hfill
	\end{equation}
	Then, we can compute 
	\begin{equation}
		\nonumber
		\hfill \underline{p}_\switch{C}(\switch{I}) := 1 - \sum_{\mu \leq k < 4 \Delta - 3} \overline{p}_\switch{C}(\switch{C}_{k}) \hfill
	\end{equation}
	in time $O(\Delta)$.
	To efficiently choose one of the switchings where $m + n + o + p = k$, observe that $p_\switch{C}(\switch{C}_{m,n,o,p})$ only depends on $k$ and not the individual choices of $m,n,o,p$, thus, it suffices to compute the exact probability for one of the switchings, and the exact number of switchings, and then follow the steps described in the proof of \autoref{lem:st1-runtime}.
	
	By \autoref{lem:st2-auxiliary} the probability that steps $6ev$ or $7ev$ are executed is at most $o(\Delta^{-1})$, so it suffices if these steps can be implemented in time $O(\Delta^2)$.
	This is possible by following the steps described towards the end of the proof of \autoref{lem:st1-runtime} with the necessary modifications.
\end{proof}

\noindent
The following shows the efficiency of \textsc{T-Gen}.

\begin{lemma}
	\label{lem:t-gen-runtime}
	\textsc{T-Gen} runs in expected time $O(M)$ for a tuple $\mathbf{D} = (\mathbf{d}, T)$ which satisfies $\Delta^{2+\epsilon} = O(M)$ for a constant $\epsilon > 0$ and $T - \Delta = \Omega(T)$.
\end{lemma}

\begin{proof}
	By \autoref{lem:tcm-temporal-loops-multi-edges} and \autoref{lem:tcm-kappa-lambda-eta} the algorithm restarts at most $O(1)$ times before finding a multigraph which satisfies the initial conditions.
	%By \autoref{lem:initial-restart} the algorithm restarts at most $O(1)$ times before finding a multigraph which satisfies the initial conditions.
	In addition, if the initial multigraph satisfies the conditions, then by \autoref{lem:st1-rejection} and \autoref{lem:st2-rejection}, the algorithm restarts at most $O(1)$ times during Stage 1 and Stage 2.
	The run time of the temporal configuration model is $O(M)$, and if the initial multigraph satisfies the conditions, then by \autoref{lem:st1-runtime} and \autoref{lem:st2-runtime}, the combined runtime of both Stage 1 and Stage 2 is $O(\Delta^3 / T) + O(\Delta^2) = O(M)$.
\end{proof}

\subsection{Proof of \autoref{thm:t-gen}}

\begin{proof}[Proof of \autoref{thm:t-gen}]
	The claim follows by \autoref{lem:t-gen-uniformity} and \autoref{lem:t-gen-runtime}.
\end{proof}

\section{Future Directions}
It would be interesting to loosen the condition on the maximum degree $\Delta$ in terms of the lifetime $T$.
Additionally, there are other parameters which could be incorporated into the model.
For instance, it could be useful to fix the number of edges assigned to each timestamp, or to specify desired multiplicities for the edges.

%\clearpage

\bibliography{tgs-bibliography}

%\clearpage

%\appendix

%\input{proofs}

\end{document}